\documentclass[a4paper,11pt,UKenglish]{amsart}
\pdfoutput=1  

\usepackage[utf8]{inputenc}
\usepackage[T1]{fontenc}

\usepackage{babel}

\usepackage{amsmath}
\usepackage{amssymb}
\usepackage{amsthm}
\usepackage{mathtools}
\usepackage{thmtools}
\usepackage{hyperref}
\usepackage[capitalize]{cleveref}
\usepackage{enumitem}
\usepackage{xspace}

\usepackage{tikz}
\usetikzlibrary{arrows}
\usetikzlibrary{decorations.markings}
\tikzstyle normalnode=[circle, draw, fill=black, inner sep=0mm, minimum width=1.5mm]
\tikzstyle smallnode=[circle, draw, fill=black, inner sep=0mm, minimum width=0.9mm]
\tikzstyle directed=[postaction={decorate,decoration={markings, mark=at position .65 with {\arrow{angle 60}}}}]
\tikzstyle reverse directed=[postaction={decorate,decoration={markings, mark=at position .65 with {\arrowreversed{angle 60};}}}]

\usepackage{biblatex}
\AtEveryBibitem{%
  \iffieldundef{doi}{}{\clearfield{url}}
  \clearfield{eprintclass}
}

\Crefname{figure}{Figure}{Figures}

\crefformat{equation}{\textup{#2(#1)#3}}
\crefrangeformat{equation}{\textup{#3(#1)#4--#5(#2)#6}}
\crefmultiformat{equation}{\textup{#2(#1)#3}}{ and~\textup{#2(#1)#3}}
{, \textup{#2(#1)#3}}{, and~\textup{#2(#1)#3}}
\crefrangemultiformat{equation}{\textup{#3(#1)#4--#5(#2)#6}}%
{ and~\textup{#3(#1)#4--#5(#2)#6}}{, \textup{#3(#1)#4--#5(#2)#6}}{, and~\textup{#3(#1)#4--#5(#2)#6}}

\Crefformat{equation}{Equation~\textup{#2(#1)#3}}
\Crefrangeformat{equation}{Equations~\textup{#3(#1)#4--#5(#2)#6}}
\Crefmultiformat{equation}{Equations~\textup{#2(#1)#3}}{ and~\textup{#2(#1)#3}}
{, \textup{#2(#1)#3}}{, and~\textup{#2(#1)#3}}
\Crefrangemultiformat{equation}{Equations~\textup{#3(#1)#4--#5(#2)#6}}%
{ and~\textup{#3(#1)#4--#5(#2)#6}}{, \textup{#3(#1)#4--#5(#2)#6}}{, and~\textup{#3(#1)#4--#5(#2)#6}}

\crefname{section}{section}{sections}

\newtheorem{theorem}{Theorem}[section]
\newtheorem{conjecture}[theorem]{Conjecture}
\newtheorem{lemma}[theorem]{Lemma}
\newtheorem{proposition}[theorem]{Proposition}
\theoremstyle{remark}
\newtheorem{remark}[theorem]{Remark}
\newtheorem{claim}[theorem]{Claim}
\newenvironment{proofofclaim}{\begin{proof}}{\end{proof}}
\newenvironment{proofsketch}{\begin{proof}[Sketch of proof]}{\end{proof}}

\newcommand{\card}[1]{|{#1}|}
\newcommand{\Nn}{\mathbb{N}}

\newcommand{\Cc}{\mathcal{C}}
\newcommand{\Dc}{\mathcal{D}}
\newcommand{\Fc}{\mathcal{F}}
\newcommand{\Hc}{\mathcal{H}}
\renewcommand{\Mc}{\mathcal{M}}
\newcommand{\Oc}{\mathcal{O}}
\newcommand{\Pc}{\mathcal{P}}
\newcommand{\Rc}{\mathcal{R}}
\newcommand{\Sc}{\mathcal{S}}
\newcommand{\Tc}{\mathcal{T}}
\newcommand{\Perm}{\mathfrak{S}}
\newcommand*\floor[1]{\left\lfloor #1 \right\rfloor}

\newcommand{\eqdef}{\coloneqq}
\newcommand{\setst}[2]{\textstyle \left\{#1 \ | \ #2 \right\}}

\newcommand{\fpt}{\ensuremath{\textrm{FPT}}\xspace}
\newcommand{\aw}{\ensuremath{\textrm{AW}[*]}\xspace}

\DeclareMathOperator{\tww}{tww}

\renewcommand{\le}{\leqslant}
\renewcommand{\ge}{\geqslant}
\renewcommand{\emptyset}{\varnothing}

\newcommand{\qlt}{\prec}

\newcommand{\qle}{\preceq}

\newcommand{\from}{\leftarrow}

\newcommand{\ot}{ot}

\title[First Order Logic and Twin-Width in Tournaments]
  {First Order Logic and Twin-Width \\ in Tournaments and Dense Oriented Graphs}

\author{Colin Geniet}
\author{Stéphan Thomassé}

\address{Institute for Basic Science (IBS), Discrete Mathematics Group, 55 Expo-ro, Yuseong-gu, Daejeon, South Korea 34126.}
\email{research@colingeniet.com}

\address{Laboratoire de l'Informatique du Parallélisme, ENS de Lyon, 46 allée d’Italie, 69364 Lyon CEDEX 07, France}
\email{stephan.thomasse@ens-lyon.fr}

\thanks{%
  Both authors were supported by the ANR projects TWIN-WIDTH (ANR-21-CE48-0014-01) and DIGRAPHS (ANR-19-CE48-0013-01).
  The first author was supported by the Institute for Basic Science (IBS-R029-C1).
}



\begin{document}
\begin{abstract}
We characterise the classes of tournaments with tractable first-order model checking.
For every hereditary class of tournaments~$\mathcal T$, first-order model checking
is either fixed parameter tractable or $\textrm{AW}[*]$-hard.
This dichotomy coincides with the fact that~$\mathcal T$ has either bounded or unbounded twin-width,
and that the growth of~$\mathcal T$ is either at most exponential or at least factorial.
From the model-theoretic point of view, we show that NIP classes of tournaments coincide with bounded twin-width.
Twin-width is also characterised by three infinite families of obstructions:
$\mathcal T$ has bounded twin-width if and only if it excludes at least one tournament from each family.
This generalises results of Bonnet et al.\ on ordered graphs.

The key for these results is a polynomial time algorithm that takes as input a tournament~$T$
and computes a linear order~$<$ on~$V(T)$ such that the twin-width of the birelation~$(T,<)$
is at most some function of the twin-width of~$T$.
Since approximating twin-width can be done in polynomial time for an ordered structure~$(T,<)$,
this provides a polynomial time approximation of twin-width for tournaments.

Our results extend to oriented graphs with stable sets of bounded size,
which may also be augmented by arbitrary binary relations.
\end{abstract}
\maketitle

\section{Introduction}
\subsection{Parameterised problems in tournaments}
Tournaments can represent the outcome of a ranking procedure on a set of candidates, which in general need not be a total order.
A well-known example is the Condorcet voting paradox, where three referees
whose preference lists are~$(A,B,C)$, $(B,C,A)$, and~$(C,A,B)$
lead to a cycle~$A \from B \from C \from A$ in the `preferred over' relation.
This interpretation leads to classical algorithmic problems when trying to choose winners.
For instance the \textsc{Dominating Set} (DS) problem on tournaments asks for a subset~$D$
that is collectively preferred to any other candidate,
in the sense that for any~$y \not\in D$, there is~$x \in D$ such that~$x$ is preferred to~$y$.
The \textsc{Feedback Vertex Set} (FVS) problem attempts to build a preference order---i.e.\ an acyclic graph---%
over the candidates by ignoring a subset of candidates.

One can consider the parameterised version of these problems, asking for a solution of size~$k$---%
they are denoted by $k$-DS, $k$-FVS.
A problem parameterised by~$k$ is \emph{fixed parameter tractable} (\fpt)
if it admits an algorithm running in time~$f(k) \cdot n^{O(1)}$ where~$n$ is the size of the input.
For instance $k$-FVS is well-known to be \fpt for tournaments (see Kumar and Lokshtanov~\cite{Kumar16}),
whereas $k$-DS is probably not \fpt.
However general tournaments may not be representative of usual instances,
for example, majority voting tournaments involving~$2r+1$ referees form a very restricted class.
A cornerstone paper by Alon et al.~\cite{alon2006dominating}, based on Vapnik-Chervonenkis dimension,
shows that $k$-DS is trivially \fpt on $(2r+1)$-majority tournaments because the size of a minimum dominating set is at most some function~$f(r)$.
This exemplifies how difficult problems can become easy on restricted classes,
and how complexity parameters (here bounded VC-dimension) can help to define classes of tournaments.

To put these questions in a much broader perspective,
observe that the previous problems can be expressed in first-order logic (FO).
The existence of a $k$-DS is described by the formula
\[ \exists x_1,x_2,\dots,x_k. \forall y. \ (y \to x_1) \lor \dots \lor (y \to x_k) \]
(assuming~$x \to x$ for all~$x$ for simplicity).
That $k$-FVS is also expressible in first-order logic is worth mentioning,
as this is only true in tournaments, and not in general graphs.
It is based on the simple remark that a tournament is acyclic if and only if
it is transitive, i.e.\ it has no directed 3-cycle---the latter condition is easily expressed in FO.
This makes $k$-DS and $k$-FVS very specific instances of the \textsc{FO Model Checking} (or FOMC) problem,
which given as input a tournament~$T$ and a first-order formula~$\phi$, asks if~$T \models \phi$.
FO model checking is \aw-hard on the class of all graphs~\cite{downey1996queries}, and thus conjectured not to be \fpt.
Using back-and-forth encodings expressible in first-order logic, one can show that the problem is just as hard in general tournaments.
Thus, one cannot hope for substantial algorithmic results on first-order model checking for arbitrary tournaments.
Instead, we study the classes of tournaments in which the problem is simple,
meaning the classes of tournaments that admit an \fpt algorithm for FO model checking parameterised by the size of the given formula.

\subsection{Main results}
We prove a dichotomy: in any class of tournaments~$\Tc$ (closed under taking subtournaments), FOMC is either \fpt or \aw-hard.
More precisely, we show that the key of this dichotomy is \emph{twin-width},
a complexity parameter introduced by Bonnet et al.~\cite{twin-width1}, defined on arbitrary binary structures.
If~$\Tc$ has bounded twin-width, then FOMC in~$\Tc$ is \fpt,
whereas it is \aw-hard when~$\Tc$ has unbounded twin-width.
This dichotomy also coincides with a model theoretic characterisation:
the class~$\Tc$ has bounded twin-width if and only if it is (monadically) NIP,
which roughly speaking means that arbitrary graphs cannot be described
by use of tournaments in~$\Tc$, and a fixed first-order formula.
In the vocabulary of~\cite{twin-width8}, this proves that tournaments are \emph{delineated}:
every subclass is NIP if and only if it has bounded twin-width,
answering a question posed in \cite{twin-width8}.
The equivalence between NIP and \fpt FO model checking also
confirms the \emph{nowhere FO dense} conjecture of Gajarsk\'y et al.~\cite{Gaj2020} in the case of tournaments.

The dichotomy for FO model checking in tournaments also coincides
with a gap in the growth function~$g_{\Tc}(n)$ of the class~$\Tc$,
defined as the number of tournaments of~$\Tc$ with~$n$ vertices counted up to isomorphism.
We show that~$\Tc$ has bounded twin-width if and only if its growth is at most~$c^n$ for some constant~$c$.
To the contrary, when twin-width is not bounded, the growth is at least $(\floor{n/2}-1)!$.
This exponential/factorial gap can be seen as a generalization of
the Marcus-Tardos theorem on permutations avoiding a fixed pattern~\cite{MarcusT04}.
It may also be compared to results of Boudabbous and Pouzet~\cite{boudabbous2010morphology}
showing that hereditary classes of tournaments have growth either at most polynomial or at least exponential---%
the growth being polynomial if and only if the class does not contain
arbitrarily large tournaments with no acyclic modules.

The following theorem summarises our results.
\begin{theorem}
  \label{thm:main-thm}
  Let~$\Tc$ be a hereditary class of tournaments.
  Assuming\footnotemark{} that $\fpt \neq \aw$, the following are equivalent:
  \footnotetext{The assumption is only used for conditions~\labelcref{item:main-thm-fpt,item:main-thm-aw-hard}.}
  \begin{enumerate}[noitemsep]
    \item $\Tc$ has bounded twin-width,
    \item \label{item:main-thm-fpt} FO model checking in~$\Tc$ is \fpt,
    \item \label{item:main-thm-aw-hard} FO model checking in~$\Tc$ is not \aw-hard,
    \item $\Tc$ does not FO interpret the class of all graphs,
    \item $\Tc$ is monadically NIP, i.e.\ does not FO transduce all graphs,
    \item the growth of $\Tc$ is at most~$c^n$ for some constant~$c$,
    \item the growth of $\Tc$ is less than~$(\floor{\frac{n}{2}}-1)!$.
  \end{enumerate}
\end{theorem}

These equivalences are completed by three minimal classes of obstructions,
characterising twin-width in terms of excluded substructures.
These excluded tournaments are encodings of arbitrary permutations.
\begin{theorem}
  \label{thm:obstructions}
  There are three obstruction classes~$\Fc_=,\Fc_\le,\Fc_\ge$ such that
  a hereditary class~$\Tc$ of tournaments has unbounded twin-width if and only if
  one of~$\Fc_=,\Fc_\le,\Fc_\ge$ is contained in~$\Tc$.
\end{theorem}

Finally, we show that there is a fixed parameter tractable algorithm
that approximates twin-width of tournaments up to some function.
\begin{restatable}{theorem}{apxalgo}
  \label{thm:approx-algo}
  There are functions~$f,g : \Nn \to \Nn$ and an algorithm that given a tournament~$T$ with twin-width~$k$,
  produces a witness that the twin-width of~$T$ is at most~$f(k)$ in \fpt time~$g(k) \cdot \card{T}^{O(1)}$.
\end{restatable}
This algorithm is crucial to obtain the \fpt FO model checking algorithm in classes with bounded twin-width.

The former three theorems generalise to oriented graphs with independence number bounded by some fixed constant.
Furthermore, the tournaments (or oriented graphs with bounded independence)
can be augmented by any fixed number of arbitrary binary relations:
the theorems still hold for the resulting class of binary relational structures.
In particular, this work generalises the analogous results of~\cite{twin-width4}
on ordered graphs and ordered binary structures: we replace linear orders with tournaments.
In these generalisations, however, the classes of obstructions for \cref{thm:obstructions}
are more numerous, and we do not give precise descriptions.

\subsection{Overview of the proof}\label{sec:proof-overview}
A fundamental idea regarding twin-width is that upper bounds on twin-width
can be witnessed by orders on vertices that exclude grid-like structures in the adjacency matrix.
This idea appears in the founding works of Guillemot and Marx~\cite{Guillemot14} and Bonnet et al.~\cite{twin-width1},
and the relation between twin-width and orders has been deeply explored in~\cite{twin-width4}.
However it is difficult to provide witnesses for lower bounds on twin-width using this approach:
one needs to somehow prove that \emph{all} orders contain grids.
To this purpose, we want to construct an order~$<$ in any given tournament~$T$,
such that when~$T$ has small twin-width, $<$ is a witness of this fact.
That is, if~$T$ has twin-width~$k$, then the bi-relation~$(T,<)$ should have twin-width at most~$f(k)$ for some function~$f$.

A tentative approach to obtain such an order is to describe it in FO logic---%
that is, to give an FO transduction that produces a total ordering on any given tournament.
Indeed, FO transductions preserve twin-width up to some function~\cite[Theorem~39]{twin-width1}.
Thus, if there is a universal transduction~$\Phi$ computing some order~$<$ on any tournament~$T$,
then we obtain~$\tww(T,<) \le f(\tww(T))$ as desired.
With such a transduction, and some additional requirements such as~$<$ being efficiently computable,
it would be straightforward to obtain our results from the known case of ordered graphs~\cite{twin-width4}.
Unfortunately, this approach fails: to transduce a total order on
an iterated lexicographic product of the 3-cycle with itself,
one needs a first-order formula with size increasing in the number of iterations~\cite{bojanczyk2022noorder}.
This class of counter-examples also has twin-width~1,
meaning that it is not even possible to transduce an order
on the class of tournaments of twin-width~$k$ for a fixed~$k$.

Instead, our approach is the following:
we design a candidate total order~$<$ on~$T$, computable in polynomial time.
If the bi-relation~$(T,<)$ has small twin-width, we are done.
On the other hand, if~$(T,<)$ has large twin-width, then its adjacency matrix w.r.t.~$<$
must contain a large high-rank grid by~\cite{twin-width4}.
We then extract a subtournament~$T' \subset T$ in which the adjacency matrix w.r.t.~$<$
still has a substantial (but logarithmically smaller) high-rank grid,
and on which~$<$ is described by an FO transduction.
This is enough to witness that~$T$ has large twin-width.
Using Ramsey arguments, we are also able to extract from~$T'$
some obstruction from one of the classes~$\Fc_=,\Fc_\le,\Fc_\ge$.

The construction of the order is simple:
we consider a binary search tree (BST) on the vertex set, i.e.\ a tree in which
the left, resp.\ right, branch of a node~$x$ consists only of in-, resp.\ out-neighbours of~$x$.
The order~$<$ is the left-to-right order on nodes of the tree.
The extraction of smaller high-rank grids corresponds to the choice of a branch~$B$,
and the restriction of~$<$ to the relevant subtournament is a lexicographic ordering according to~$B$.

To summarise, the crucial property of BST orders is the following.
\begin{restatable}{lemma}{bsttww}
  \label{lem:BST-twinwidth}
  There is a function~$f$ such that for any tournament~$T$ and any BST order~$<$,
  \[ \tww(T,<) \le f(\tww(T)). \]
\end{restatable}
\Cref{lem:BST-twinwidth} implies~\cref{thm:approx-algo}:
to approximate the twin-width of~$T$, it suffices to compute any BST order---this takes polynomial time---and then apply the approximation algorithm for ordered structures~\cite[Theorem~7]{twin-width4}---this is fixed parameter tractable.
This last algorithm produces either a contraction sequence
(which is valid for~$(T,<)$ and a fortiori for~$T$)
or a witness that~$(T,<)$ has large twin-width,
which in turn implies that~$T$ has large twin-width by \cref{lem:BST-twinwidth}.

Our main technical result is about extracting the three classes of obstructions to twin-width~$\Fc_=,\Fc_\le,\Fc_\ge$, consisting of encodings of arbitrary permutations into tournaments.
See \cref{sec:obstructions} for the definition.
\begin{theorem}
  \label{thm:forbidden-tournaments-intro}
  Let~$\Tc$ be a hereditary class of tournaments such that for any~$k$,
  there is a tournament~$T \in \Tc$ and a BST ordering~$<$ of~$T$ with $\tww(T,<) > k$.
  Then~$\Tc$ contains one of the classes~$\Fc_=,\Fc_\le,\Fc_\ge$ as a subclass.
\end{theorem}

We also prove that the classes~$\Fc_=,\Fc_\le,\Fc_\ge$ are complex
in all the senses considered by \cref{thm:main-thm}, that is
\begin{restatable}{theorem}{obstrhard}
  \label{thm:forbidden-tournaments-bad}
  For each~$R \in \{=,\le,\ge\}$, the class~$\Fc_R$
  \begin{enumerate}[noitemsep]
    \item \label{item:unbounded-tww} has unbounded twin-width,
    \item \label{item:factorial} contains at least~$(\floor{\frac{n}{2}} - 1)!$ tournaments
      on~$n$ vertices counted up to isomorphism,
    \item \label{item:not-small} contains at least~$(\floor{\frac{n}{2}} - 1)!  \cdot n!$ tournaments on vertex
      set~$\{1,\dots,n\}$ counted up to equality,
    \item \label{item:independent} efficiently interprets the class of all graphs,
    \item \label{item:aw-hard} and has an \aw-hard FO model checking problem.
  \end{enumerate}
\end{restatable}

\Cref{thm:forbidden-tournaments-intro,thm:forbidden-tournaments-bad} together imply \cref{thm:obstructions}.
By applying them to the class of tournaments with twin-width at most~$k$, they also imply \cref{lem:BST-twinwidth}:
this class cannot contain any of~$\Fc_=,\Fc_\le,\Fc_\ge$,
hence its tournaments must still have bounded twin-width when paired with BST orders.
Finally, \cref{thm:obstructions,thm:forbidden-tournaments-bad} directly imply that
if~$\Tc$ is a hereditary class with unbounded twin-width,
then~$\Tc$ satisfies none of the conditions of \cref{thm:main-thm}.
The remaining implications of \cref{thm:main-thm}---%
that is, when~$\Tc$ has bounded twin-width, all other conditions hold---%
follow from known results on twin-width.
By \cite[Theorem~1]{twin-width1}, FO model checking has an \fpt algorithm
when a witness of bounded twin-width is given.
Combined with \cref{thm:approx-algo}, we obtain an \fpt algorithm
for classes of tournaments with bounded twin-width.
By \cite[Theorem~39]{twin-width1}, a class of structures with bounded twin-width cannot transduce all graphs.
Finally, by \cite[Corollary~7.3]{tww-perm}, a class of structures with bounded twin-width
contains at most~$c^n$ structures on~$n$ vertices up to isomorphism, for some constant~$c$ depending on the class.

\subsection{Context and related works}
An important question is how twin-width compares with other classical complexity measures for tournaments.
Bounded twin-width implies bounded Vapnik-Chervonenkis dimension:
indeed unbounded VC-dimension implies all possible bipartite tournaments appearing as subgraphs,
which is against single-exponential growth, and thus twin-width is unbounded.
The notion of \emph{cutwidth} was introduced
by Chudnovsky, Fradkin and Seymour~\cite{Chud2012} to study tournament immersions.
The vertex ordering used to certify that cutwidth is bounded
can be shown to exclude grids, meaning that it is a witness of bounded twin-width.
Another parameter, closely related to subdivisions in tournaments,
is \emph{pathwidth}, studied by Fradkin and Seymour~\cite{FRADKIN2013374}.
Bounded pathwidth of tournaments implies bounded cliquewidth,
which in turn also implies bounded twin-width, see~\cite{twin-width1}.
Thus, we have the following chain of inclusions of complexity parameters
(if a parameter is bounded, all the ones listed after are also bounded):
cutwidth, pathwidth, cliquewidth, twin-width, and VC-dimension.
For more on the subject, we recommend reading the work of Fomin and Pilipczuk~\cite{FOMIN201978,pilipczuk2013tournaments}.

Our results are a direct generalisation of those of~\cite{twin-width4}, from ordered graphs to tournaments.
This generalisation is non-trivial: as explained in \cref{sec:proof-overview},
one cannot FO transduce orders from arbitrary tournaments---%
which would be the obvious way to reduce the problem from tournaments to ordered graphs.
On the other hand, there is a natural way to form a tournament~$T$ from an ordered graph~$(G,<)$,
namely $x \to y$ is a directed edge in~$T$ if and only if~$x<y$ and~$xy$ is an edge of~$G$, or~$y<x$ and~$xy$ is not an edge of~$G$.
This gives a way to create tournaments with bounded twin-width:
starting with any graph~$G$ with bounded twin-width,
consider a total order~$<$ on its vertices that is a witness of its twin-width,
and interpret~$T$ from~$(G,<)$ as before.
This paper can be seen as the reverse operation,
that is a way of producing an order~$<$ from a tournament~$T$.

The binary search tree we use to order tournaments
correspond to the \textsc{KwikSort} algorithm of Ailon, Charikar and Newman
for approximating the minimum feedback arc set~\cite{newman2008FAS}.
The only difference is that their result requires the BST to be randomly chosen,
whereas arbitrary BST provide approximations of twin-width.

The generalization to oriented graphs with bounded independence number
contains in particular partially ordered sets with bounded width (i.e.\ antichains of bounded size).
These classes in fact always have bounded twin-width, see Balab\'{a}n and Hlinen\'{y}~\cite{balaban2021posets},
and before that were already known to be tractable with respect to FO model checking,
see Gajarsk\'{y} et al.~\cite{gajarsky2015posets}.

Finally let us mention another notable result on tournaments of bounded twin-width:
the polynomial time isomorphism test of Grohe and Neuen~\cite{grohe2024isomorphism}.
Their algorithm is remarkable in that it leverages twin-width
without needing or manipulating a contraction sequence witness of twin-width.
The techniques involved are very different from ours.

\subsection{Organisation of the paper}
\Cref{sec:prelim} summarises definitions and notations used in this paper.
In \cref{sec:obstructions} we define the classes~$\Fc_=,\Fc_\le,\Fc_\ge$ of
obstructions to twin-width in tournaments, and prove \cref{thm:forbidden-tournaments-bad}.
\Cref{sec:bst} introduces binary search trees, and proves a crucial lemma,
which from a partition into intervals of a BST order, extracts some FO definable ordered substructure.
\Cref{sec:approx} combines this key lemma with results of~\cite{twin-width4} to prove~\cref{lem:BST-twinwidth}:
from a witness that the tournament~$T$ augmented by a BST order has large twin-width,
the key lemma extracts a witness that~$T$ itself has large twin-width.
\Cref{sec:extraction} refines this argument into a proof of~\cref{thm:forbidden-tournaments-intro} using Ramsey results.
Finally, \cref{sec:generalisation} explains how our results generalise
to oriented graphs with bounded independence number, and to binary relational structures.

\section{Preliminaries}\label{sec:prelim}
This section summarises the notions and notations used in this work.
For $n \in \Nn$, we denote by~$[n]$ the interval of integers~$\{1,\dots,n\}$.
All the graphs, matrices, and relational structures considered in this work are finite.

\subsection{Graphs, Oriented graphs, Tournaments}
A \emph{directed} graph $D$ consists of a vertex set~$V(G)$,
and a set~$E(G)$ of \emph{directed edges} (also called \emph{arcs}), which are \emph{ordered} pairs of vertices.
A directed edge~$(u,v)$ is simply denoted by~$uv$, or sometimes~$u \to v$ to insist on the orientation.
In this work, we do not consider loops, i.e.\ edges of the form~$v \to v$.
The \emph{underlying graph} of~$D$ is the undirected graph with the same vertices,
and an edge~$uv$ whenever either~$uv$ or~$vu$ is an edge of~$D$.
A \emph{digon} in a directed graph is a pair of vertices~$u,v$ with both edges~$u \to v$ and~$v \to u$.
An \emph{oriented} graph is a directed graph without digons.

A \emph{tournament} is an oriented graph whose underlying graph is complete,
or equivalently it is a choice of orientation~$u \to v$ or~$v \to u$ (but not both) for every pair of vertices.
A tournament is \emph{transitive} if it contains no directed cycle.
A transitive tournament represents some order~$<$ on vertices: there is an edge~$u \to v$ if and only if~$u<v$.
In a directed graph~$D$, a subset~$X \subset V(D)$ that induces a transitive tournament is also called a \emph{chain}.

A (directed) graph~$H$ is an \emph{induced subgraph} of~$G$ if~$V(H) \subset V(G)$,
and the edges of~$H$ are exactly the edges of~$G$ that are between vertices of~$V(H)$.
If~$X = V(H)$, we also say that~$H$ is the subgraph of~$G$ induced by~$X$, denoted~$H = G[X]$.
A \emph{class} of graphs is a family of graphs closed under isomorphism.
A class is \emph{hereditary} if it is also closed under taking induced subgraphs.

A subset~$S$ of a (directed) graph is called \emph{independent} or \emph{stable}
if no edge joins two vertices of~$S$.
The independence number~$\alpha(G)$ is the maximum size of an independent set in~$G$.
For a vertex~$x$ in a directed graph~$D$, the \emph{out-neighbourhood} is $N^+(x) = \setst{y}{x \to y}$,
and the \emph{in-neighbourhood} is $N^-(x) = \setst{y}{y \to x}$.
Their union is the \emph{complete neighbourhood} $N(x) = N^+(x) \cup N^-(x)$.

\subsection{Relational structures}
Relational structures are a natural generalisation of graphs.
A relational \emph{signature} is a finite set~$\Sigma$ of \emph{relation symbols}~$R$,
each with an associated arity~$r \in \Nn$.
A relational structure~$S$ over~$\Sigma$, or simply $\Sigma$-structure,
consists of a domain~$A$, and for each relation symbol~$R \in \Sigma$ of arity~$r$,
an interpretation of~$R$ as~$R^S \subseteq A^r$.
One may think of~$A$ as a set of vertices, and each~$R^S$ as a set of hyperedges.
The notions of induced substructure, hereditary classes, etc.\ generalise to relational structures in the obvious way.

We will only consider \emph{binary} relational structures,
i.e.\ over signatures where all symbols have arity~2.
Directed graphs are binary structures with a single relation for edges.

An \emph{ordered structure}~$S$ is a structure over a relation~$\Sigma$ with a special symbol,
typically denoted~`$<$', whose interpretation~$<^S$ is a strict total order on the domain of~$S$.

\subsection{Orders, Quasi-orders}
A \emph{quasi-order}~$\qle$ is a reflexive and transitive binary relation.
Any quasi-order induces an equivalence relation~$x \sim y$ defined as~$x \qle y \land y \qle x$,
which is trivial if and only if~$\qle$ is an order.
One can also describe a quasi-order by giving this equivalence relation,
and the order on the equivalence classes.
The strict component of the quasi-order is defined by $x \qlt y$ if and only if~$x \qle y$ and~$y \not\qle x$.
The quasi-order is \emph{total} if for all~$x,y$, either~$x \qle y$ or~$y \qle x$.

An \emph{interval} of a quasi-order~$\qle$ is a set of the form~$\setst{z}{x \qle z \qle y}$
for some~$x,y$, called \emph{endpoints} of the interval.
Remark that an interval is always a union of equivalence classes of~$\sim$.

\subsection{Matrices}
As is common in the context of permutations and patterns,
we use the convention that matrix coordinate are oriented as in the euclidean plane.
That is, coordinate~$(i,j)$ means column~$i$ and row~$j$, with the bottom-most row being the smallest.

A matrix is seen as a map~$M: C \times R \to \Gamma$,
where~$C,R$ are \emph{ordered sets} of columns and rows of the matrix, and~$\Gamma$ and its alphabet.
The order of rows and columns is very important in the context of twin-width.
Usually, the alphabet will simply be~$\{0,1\}$.
The value~$M(x,y)$ is called the \emph{entry} at position~$(x,y)$,
and an $a$-entry means an entry whose value is~$a$.
A submatrix of~$M$ is the restriction of~$M$ to some subsets of rows and columns.
The notion is comparable to that of induced substructure, thinking of rows and columns as vertices,
and of the content of the matrix as a relation.
Accordingly, a class of matrices is called \emph{hereditary} if it is closed under submatrices.

Given a (directed) graph~$G$ and a total order~$<$ on~$V(G)$,
the adjacency matrix~$A_{(G,<)}$ is a 0,1-matrix with rows and columns~$V(G)$,
with a~`1' at position~$(u,v)$ if and only if~$uv$ is an edge of~$G$.
Note that we do not use the common alternative definition for oriented graphs which places~`$-1$' at position~$(u,v)$ when~$u \from v$ instead is an edge.
Adjacency matrices generalise easily to binary relational structures over any signature~$\Sigma$,
the alphabet being~$\{0,1\}^\Sigma$ in this case.

A \emph{division}~$\Dc$ of a matrix consist of partitions~$\Rc,\Cc$
of its rows and its columns respectively into \emph{intervals}.
It is a $k$-division if the partitions have~$k$ parts each.
A \emph{cell} of the division is the submatrix induced by~$X \times Y$ for some~$X \in \Rc, Y \in \Cc$.
A $k$-\emph{grid} in a 0,1-matrix is a division in which every cell contains a `1'-entry.

\subsection{Permutations}\label{sec:permutations}
We denote by~$\Perm_n$ the group of permutations on~$n$ elements.
The \emph{permutation matrix}~$M_\sigma$ of~$\sigma \in \Perm_n$ is the sparse $n \times n$ matrix
with a~`1' at position~$(i,j)$ if and only if~$j = \sigma(i)$.

A~permutation~$\tau$ is a \emph{pattern} of~$\sigma$ if $M_\tau$ is a submatrix of~$M_\sigma$.
Any pattern~$\tau$ of~$\sigma$ is obtained by restricting~$M_\sigma$ to columns in~$X$ and rows in~$\sigma(X)$ for some subset~$X$ of the domain of~$\sigma$.
We then say that~$\tau$ is the pattern of~$\sigma$ \emph{induced} by~$X$.

We say that~$\sigma$ contains a $k$-grid if $M_\sigma$ contains a $k$-grid.
When this is the case, any permutation in~$\Perm_k$ is a pattern of~$\sigma$.
For example, the permutation~$\sigma$ on $k^2$ elements
defined by~$\sigma(ki + j + 1) = kj + i + 1$ for any~$0 \le i,j < k$ contains a $k$-grid.
Grids in permutations are deeply related to pattern-closed classes of permutations,
which are precursors of the work on twin-width,
see Marcus and Tardos~\cite{MarcusT04} and Guillemot and Marx~\cite{Guillemot14}.

\begin{lemma}\label{lem:perm-grid-sym}
  If the permutation~$\sigma \in \Perm_n$ contains a $k$-grid, then so do its inverse~$\sigma^{-1}$, and the reverse $i \mapsto n - \sigma(i) + 1$ and $i \mapsto \sigma(n-i+1)$.
\end{lemma}
\begin{proof}
  These operations correspond to transposing or mirroring the permutation matrix respectively, which clearly preserves the presence of $k$-grids.
\end{proof}

\begin{lemma}\label{lem:perm-grid-minus-one}
  Consider a permutation $\sigma \in \Perm_n$ with a $(k+1)$-grid, and $x \in [n]$ any element of its domain.
  Then the pattern of~$\sigma$ induced by~$[n] \setminus \{x\}$ contains a $k$-grid.
\end{lemma}
\begin{proof}
  In the matrix~$M_\sigma$, the pattern induced by~$[n] \setminus \{x\}$ corresponds to removing column~$x$ and row~$\sigma(x)$.
  Removing a single column and single row of the matrix affects only one interval of columns and one interval of rows in the $(k+1)$-grid, leaving a $k$-grid.
\end{proof}

\begin{lemma}\label{lem:perm-grid-remove-chain}
  Consider a permutation $\sigma \in \Perm_n$ with a $2k$-grid,
  and $X \subset [n]$ such that the pattern of~$\sigma$ induced by~$X$ is either increasing or decreasing.
  Then the pattern of~$\sigma$ induced by~$[n] \setminus X$ contains a $k$-grid.
\end{lemma}
\begin{proof}
  Call $R_1 < \dots < R_{2k}$ and $C_1 < \dots < C_{2k}$ the row intervals and column intervals in the $2k$-grid of~$\sigma$.
  Merge them two-by-two, as $R'_i = R_{2i-1} \cup R_i$ and $C'_i = C_{2i-1} \cup C_i$.
  Thus each cell $R'_i \times C'_j$ contains a 2-grid.

  Say that a 1-entry in~$M_\sigma$ \emph{comes from~$X$} if it is at position $(i,\sigma(i))$ for some~$i \in X$.
  Since the pattern induced by~$X$ is increasing or decreasing, it cannot contain a 2-grid.
  Thus each cell $R'_i \times C'_j$ contains a 1-entry that does not come from~$X$.
  Therefore $R'_1,\dots,R'_k$ and $C'_1,\dots,C'_k$ give a $k$-grid after removing the 1-entries from~$X$, hence a $k$-grid in the matrix of the pattern induced by $[n] \setminus X$ as desired.
\end{proof}

A permutation can be represented as a \emph{bi-order}, i.e.\ the superposition of two total orders.
Precisely, for~$\sigma \in \Perm_n$, the structure~$\Oc_\sigma$ has domain~$[n]$,
and has for relations the natural order~$<$, and the permuted order~$<_\sigma$
defined as~$i <_\sigma j$ if and only if~$\sigma(i) < \sigma(j)$.
Any bi-order is isomorphic to some~$\Oc_\sigma$.
Remark that~$\tau$ is a pattern of~$\sigma$ if and only if
$\Oc_\tau$ is isomorphic to an induced substructure of~$\Oc_\sigma$.
We write~$\Oc_\Perm$ for the class of all finite bi-orders.

For any numbers~$x,y$ let the \emph{order type} $\ot(x,y)$ be~1 if~$x<y$, $-1$ if~$x > y$, and~0 if~$x = y$.
We also write~$\ot_\sigma(x,y)$ for~$\ot(\sigma(x),\sigma(y))$,
the order type with regards to the permuted order~$<_\sigma$.
In a bi-order~$\Oc_\sigma = ([n],<,<_\sigma)$,
consider a colouring of pairs~$\lambda : [n]^2 \to \Gamma$ for some finite alphabet~$\Gamma$.
We say that this colouring \emph{depends only on the orders~$<,<_\sigma$}
if there is a $\eta : \{-1,0,1\}^2 \to \Gamma$ such that~$\lambda$ factors as
\[ \lambda(x,y) = \eta(\ot(x,y), \ot_\sigma(x,y)) \]
For example, seeing the edges of a graph on~$[n]$ as a colouring,
the permutation graph with an edge~$xy$ if only if $(x < y \land x >_\sigma y)$ or $(x > y \land x <_\sigma y)$
is one which only depends on~$<$ and~$<_\sigma$.

The following result of Ramsey theory, stated as such in~\cite[Lemma~18]{twin-width4},
will be used extensively to extract the minimal classes of obstructions in tournaments with large twin-width.
It is a corollary of the \emph{product Ramsey theorem},
see~\cite[Proposition~3.3]{bodirsky2015ramsey} or~\cite[page~97]{graham1991ramsey}.
\begin{lemma}
  \label{lem:perm-ramsey}
  For any permutation~$\sigma$ and any finite alphabet~$\Gamma$, there is a permutation~$\tau$ satisfying the following.
  Write the biorder~$\Oc_\tau$ as~$([N],<,<_\tau)$. For any colouring of pairs~$\lambda : [N]^2 \to \Gamma$,
  there exists~$X \subset N$ such that the substructure~$\Oc_\tau[X]$ is isomorphic to~$\Oc_\sigma$,
  and the colouring~$\lambda$ restricted to~$X$ only depends on the orders~$<,<_\tau$.
\end{lemma}

When the colouring~$\lambda$ represents the edges of a tournament on~$[n]$,
that is $\lambda(x,y) = 1$ if and only if~$x \to y$ is an edge,
the `depends only on~$<,<_\sigma$' conditions becomes particularly useful.
\begin{lemma}
  \label{lem:tournament-order}
  Let~$T$ be a tournament and~$<_1,<_2$ two orders on~$V(T)$
  such that the direction of edges of~$T$ only depends on~$<_1,<_2$.
  Then~$T$ is transitive, and the direction of edges coincides
  with one of~$<_1$, $<_2$, the reverse of~$<_1$, or the reverse of~$<_2$.
\end{lemma}
\begin{proof}
  Let~$\eta : \{-1,0,1\}^2 \to \{0,1\}$ be a function such that such that~$\lambda$ factors as
  $\lambda(x,y) = \eta(\ot(x,y), \ot_\sigma(x,y))$.
  Since~$T$ is a tournament, it is then simple to verify that~$\eta(-1,-1) = 1 - \eta(1,1)$
  and~$\eta(-1,1) = 1 - \eta(1,-1)$.
  It follows that the value of~$\eta$ only actually depends on one of the two coordinates.
  It is then simple to verify that~$T$ is transitive, with either the same order
  or the reverse of the order corresponding to this coordinate.
\end{proof}

\subsection{Twin-width}
We recall the definition of twin-width for the sake of completeness,
however we will not use it directly, relying on higher level results instead.

In a binary relational structure~$S$, two disjoint subsets~$X,Y$ of the domain are \emph{homogeneous}
if for all relations~$R$ of~$S$, if there are~$x \in X,y \in Y$ with~$x R y$,
then for all~$x' \in X, y' \in Y$, $x' R y'$, and symmetrically.
Thus, the relations between~$X$ and~$Y$ do not distinguish any vertex.
For example, in a graph, $X,Y$ are homogeneous if and only if
the bipartite graph between~$X$ and~$Y$ is either empty or complete;
and in a tournament, $X$ and~$Y$ are homogeneous if and only if
all edges are oriented from~$X$ to~$Y$, or all from~$Y$ to~$X$.

A \emph{contraction sequence} for a structure~$S$ is a sequence
$\Pc_n,\dots,\Pc_1$ of partitions of the domain of~$S$,
where~$\Pc_n$ is the partition into singletons, $\Pc_1$ is the partition in a single part,
and~$\Pc_i$ is obtained by merging two parts of~$\Pc_i$.
For a given part~$X \in \Pc_i$, the \emph{error degree} of~$X$ in~$(S,\Pc_i)$
is the number of parts~$Y \in \Pc_i \setminus \{X\}$ that are not homogeneous to~$X$.
The \emph{width} of the contraction sequence is the maximum error degree of~$X$ in~$(S,\Pc_i)$
over all choices of~$i \in [n]$ and~$X \in \Pc_i$.
Finally, the \emph{twin-width} of~$S$, denoted~$\tww(S)$, is the minimum width of a contraction sequence for~$S$.

Twin-width is characterised by grid-like structures in adjacency matrices.
Recall that a division of a matrix is a partition of rows and columns into intervals.
Then, a \emph{rank-$k$ division} is a $k$-division in which every cell has rank at least~$k$.
Bonnet et al.\ proved
\begin{theorem}[{\cite[Theorem~7]{twin-width4}}]
  \label{thm:grid-theorem}
  There are functions~$f,g$ such that for any graph (or binary structure)~$G$ and any order~$<$ on~$V(G)$,
  \begin{itemize}
    \item if $\tww(G,<) \ge f(k)$, then the matrix~$A_{(G,<)}$ has a rank-$k$ division, and
    \item if the matrix~$A_{(G,<)}$ has a rank-$g(k)$ division, then $\tww(G,<) \ge k$.
  \end{itemize}
  Furthermore there is an \fpt algorithm which given~$G$, $<$, and~$k$,
  finds either a rank-$k$ division in~$A_{(G,<)}$ or a contraction sequence of width~$f(k)$ for~$(G,<)$.
\end{theorem}

\subsection{First-order logic}
We assume general familiarity with first-order logic on graphs and relational structures.
Recall from the introduction that we are primarily interesting in the \textsc{First-Order Model Checking}
algorithmic problem, which asks given as input a structure~$S$ and a first-order formula~$\phi$,
if~$\phi$ is satisfied by~$S$, which is denoted by~$S \models \phi$.
We consider the complexity of this problem parameterised by the size~$\card{\phi}$.
In general, this problem is \aw-complete~\cite[Theorem~1]{downey1996queries}.
In fact, the proof shows more precisely that
\begin{theorem}[{\cite[Proposition~2]{downey1996queries}}]
  \label{thm:FO-aw-hard}
  \textsc{First-Order Model Checking} on the class of all graphs is \aw-complete.
\end{theorem}
On classes of structures with bounded twin-width, FO model checking is \fpt
as long as a witness of twin-width is given.
\begin{theorem}[{\cite[Theorem~1]{twin-width1}}]
  \label{thm:FO-fpt}
  There is an algorithm that, given a binary structure~$S$ on~$n$ vertices,
  a contraction sequence of width~$k$ for~$S$, and an FO formula~$\phi$,
  decides if~$S \models \phi$ in time~$f(k,\phi) \cdot n$.
\end{theorem}

\emph{Interpretations} and \emph{transductions} are ways to describe
transformations of structures using logical formulæ.
For two relational signatures~$\Sigma,\Delta$, an \emph{FO interpretation}~$\Phi$ from~$\Sigma$ to~$\Delta$
consists of, for each relation~$R \in \Delta$ of arity~$r$,
an FO formula~$\phi_R(x_1,\dots,x_r)$ over the language~$\Sigma$,
and one last formula~$\phi_{dom}(x)$ again over~$\Sigma$.
Given a $\Sigma$-structure~$S$, the result~$\Phi(S)$ of the interpretation is the $\Delta$-structure obtained by
\begin{itemize}[noitemsep]
  \item choosing the same domain as~$S$,
  \item interpreting each relation~$R \in \Delta$ as the set of tuples $(v_1,\dots,v_r)$ such that~$S \models \phi_R(v_1,\dots,v_r)$,
  \item and finally taking the substructure induced by~$\setst{v}{S \models \phi_{dom}(v)}$.
\end{itemize}
For instance, the square of a graph~$G$ is defined as the graph with vertices~$V(G)$,
and in which~$x,y$ are adjacent if and only if their distance in~$G$ is at most~2.
This is a first-order interpretation in which the edge relation is interpreted by the formula
\[ \phi(x,y) \ = \ x \neq y \land \big(E(x,y) \lor \exists z. \ E(x,z) \land E(z,y) \big) \]
where~$E(x,y)$ denotes adjacency.
The domain formula is trivial since we do not wish to delete vertices in this interpretation.
It is well-known that the composition of two FO interpretations can itself be written as an FO interpretation.

\emph{Transductions} generalise interpretation with a non-deterministic colouring step.
Consider relational signatures~$\Sigma$ and~$\Sigma^+$,
where~$\Sigma^+$ consists of~$\Sigma$ augmented by~$r$ new \emph{unary} relations~$C_1,\dots,C_r$.
If~$S$ is a structure over the signature~$\Sigma$,
we denote by~$S^+$ the set of structures over~$\Sigma^+$ whose domain
and interpretation of~$\Sigma$ coincide with that of~$S$.
Thus, a~$T \in S^+$ is uniquely described by the interpretation of $C_1,\dots,C_r$ within~$V(S)$.
A FO transduction~$\Phi : \Sigma \to \Delta$ is described by
the choice of~$\Sigma^+$ augmenting~$\Sigma$ with unary relations,
together with an FO interpretation $\Phi_I$ from~$\Sigma^+$ from~$\Delta$.
The result of the transduction~$\Phi$ on a $\Sigma$-structure~$S$ is then the set of structures
\[ \Phi(S) = \setst{\Phi_I(T)}{T \in S^+} \]
Additional operations such as duplication of vertices are often allowed in transductions,
but these will not be needed in this work.
Like interpretations, transductions can be composed.

Given classes~$\Cc,\Dc$ of relational structures,
we say that~$\Cc$ \emph{interprets} (resp.\ \emph{transduces})~$\Dc$
if there is an FO interpretation (resp.\ transduction)~$\Phi$ such that $\Phi(\Cc) \supseteq \Dc$.
The interpretation~$\Phi$ is allowed to also output structures that are not in~$\Dc$.
We furthermore say that~$\Cc$ \emph{efficiently interprets}~$\Dc$ if there is also
an algorithm that given as input~$D \in \Dc$, finds in polynomial time some~$C \in \Cc$ such that~$\Phi(C) = D$.
It is straightforward to show that this additional condition gives an \fpt reduction for model checking.
\begin{lemma}
  \label{lem:interpretation-reduction}
  If~$\Cc$ efficiently interprets~$\Dc$, then there is an \fpt reduction
  from \textsc{FO Model Checking} on~$\Dc$ to \textsc{FO Model Checking} on~$\Cc$.
\end{lemma}

Recall that~$\Oc_\Perm$ denotes the class of bi-orders, which are encodings of permutations.
The following is a folklore result, see e.g.~\cite[Lemma~12]{twin-width4} for a very similar claim.
\begin{lemma}
  \label{lem:perm-inter-graph}
  The class~$\Oc_\Perm$ of bi-orders efficiently interprets the class of all graphs.
\end{lemma}
Thus, using \cref{lem:interpretation-reduction,thm:FO-aw-hard},
\textsc{FO Model Checking} on~$\Oc_\Perm$ is \aw-complete.

FO transductions also preserve twin-width, up to some function.
\begin{theorem}[{\cite[Theorem~39]{twin-width1}}]
  \label{thm:tww-transduction}
  If~$\Sc$ is a class of binary structures with bounded twin-width and~$\Phi$ is an FO transduction defined on~$\Sc$,
  then~$\Phi(\Sc)$ also has bounded twin-width.
\end{theorem}

A class of structures~$\Sc$ is said to be \emph{monadically NIP}
if~$\Sc$ does not transduce the class of all graphs.
Since the class of all graphs does not have bounded twin-width,
\cref{thm:tww-transduction} implies that classes with bounded twin-width are monadically NIP.
The weaker notion of (non-monadically) NIP also exists, with a slightly more technical definition.
Braunfeld and Laskowski recently proved that NIP and monadically NIP
are equivalent notions for hereditary classes of structures~\cite{braunfeld2022existential},
thus the distinction is irrelevant in this work.

The following remarkable conjecture of Gajarsky et al.\ relates NIP classes and model checking.
\begin{conjecture}[{\cite[Conjecture~8.2]{Gaj2020}}]
  \label{conj:nip-model-checking}
  A hereditary class of graphs has an \fpt algorithm for first-order model checking if and only if it is NIP.
\end{conjecture}
The first implication of \cref{conj:nip-model-checking} was recently proved
by Dreier, Mählmann, and Toruńczyk~\cite{dreier2024dependent}, under the assumption $\aw \neq \fpt$:
they show that if~$\Cc$ is hereditary and transduces the class of all graphs,
then it in fact efficiently interprets all graphs,
hence model checking in~$\Cc$ is \aw-complete by \cref{lem:interpretation-reduction}.
The converse implication (finding a model checking algorithm in NIP classes) is open.

Our main result confirms \cref{conj:nip-model-checking} for tournaments.

\subsection{Enumerative combinatorics}
A class~$\Sc$ of graphs (or binary relational structures) is \emph{small}
if there exists~$c$ such that~$\Sc$ contains at most~$c^n \cdot n!$ structures on the vertex set~$[n]$.
For instance, the class of trees is small, and more generally proper minor closed classes of graphs
are small as shown by Norine et al.~\cite{NORINE2006754}.
This was further generalised to classes of bounded twin-width by Bonnet et al.
\begin{theorem}[{\cite[Theorem~2.4]{twin-width2}}]
  \label{thm:small}
  Classes of structures with bounded twin-width are small.
\end{theorem}
The reverse implication does not hold for graphs:
there are small hereditary classes of graphs with unbounded twin-width~\cite{twin-width7}.
These counterexamples have bounded maximum degree.
One of the main results of~\cite{twin-width4} is that for classes of ordered graphs,
bounded twin-width is equivalent to being small.
We generalise it to tournaments.

One may also count structures up to isomorphism
instead of counting labelled structures as in the definition of small classes.
A strengthening of \cref{thm:small} when counting up to isomorphism was proved in~\cite{tww-perm}.
\begin{theorem}[{\cite[Corollary~7.3]{tww-perm}}]
  \label{thm:exponential}
  For any class~$\Cc$ of structures with bounded twin-width,
  there is a constant~$c$ such that~$\Cc$ has at most~$c^n$ structures on~$n$ vertices counted up to isomorphism.
\end{theorem}

\section{Forbidden classes of tournaments}\label{sec:obstructions}
This section defines the three minimal classes~$\Fc_=$, $\Fc_\le$, and~$\Fc_\ge$
of obstructions to twin-width in tournaments.
Each of them corresponds to some encoding of the class of all permutations.
For~$R \in \{=,\le,\ge\}$ and any permutation~$\sigma$, we will define a tournament~$\Fc_R(\sigma)$.
The class~$\Fc_R$ is the hereditary closure of all~$\Fc_R(\sigma)$.

Let~$R \in \{=,\le,\ge\}$, and let~$\sigma \in \Perm_n$ be a permutation on~$n$ elements.
The tournament~$\Fc_R(\sigma)$ consists of~$2n$ vertices, called~$x_1,\dots,x_n,y_1,\dots,y_n$.
Let $X = \{x_1,\dots,x_n\}$ and $Y = \{x_1,\dots,y_n\}$.
Each of~$X,Y$ is a chain under the natural order, i.e.\ there is an edge
from~$x_i$ to~$x_j$, resp.\ from~$y_i$ to~$y_j$, if and only if~$i < j$.
The edges between~$X$ and~$Y$ encode~$\sigma$ in a way specified by the relation~$R$:
there is an edge oriented from~$x_i$ to~$y_j$ if and only if $i \, R \, \sigma^{-1}(j)$.
See \cref{fig:obstructions} for an example.

\begin{figure}
  \begin{center}
    \begin{tikzpicture}[>=angle 60]

      \def\n{5}
      \foreach \R/\s/\lbl in {E/0/=,L/4/\le,G/8/\ge}{
        \foreach \i in {1,...,\n}{
          \node[normalnode] (\R{}x\i) at (\s+0,\i) [label=left:$x_\i$] {};
          \node[normalnode] (\R{}y\i) at (\s+1.8,\i) [label=right:$y_\i$] {};
        }
        \node (\R{}label) at (\s+1,0) {$\Fc_\lbl(\sigma)$};
        \draw[thick, ->] (\R{}x1) -- (\R{}x\n);
        \draw[thick, ->] (\R{}y1) -- (\R{}y\n);
      }
      \foreach \i/\si in {1/3,2/1,3/4,4/5,5/2}{
        \foreach \k in {1,...,\n}{
          \ifnum \k=\i
            \draw[directed] (E{}x\k) -- (E{}y\si);
          \fi
          \ifnum \k>\i \else
            \draw[directed] (L{}x\k) -- (L{}y\si);
          \fi
          \ifnum \k<\i \else
            \draw[directed] (G{}x\k) -- (G{}y\si);
          \fi
        }
      }
    \end{tikzpicture}
  \end{center}
  \caption{%
    The three classes of obstructions to twin-width in tournaments.
    For readability, edges oriented from some~$y_j$ to some~$x_i$ have been omitted.
    Each class consists of some encoding of the class of all permutations,
    represented here with the permutation $\sigma = 31452$.
  }
  \label{fig:obstructions}
\end{figure}
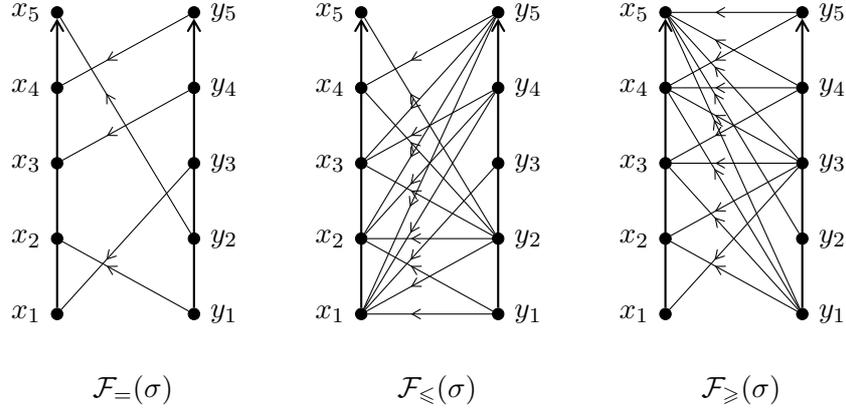

Thus in~$\Fc_=(\sigma)$ the edges oriented from~$X$ to~$Y$ form a matching which encodes~$\sigma$.
In~$\Fc_\le(\sigma)$ and~$\Fc_\ge(\sigma)$, these edges form a half-graph which orders~$X$ and~$Y$
by inclusion of neighbourhoods, so that the order on~$X$ is the natural one, and the order on~$Y$ encodes~$\sigma$.
Precisely, in~$\Fc_\ge(\sigma)$, for any $i,j \in [n]$, we have
\begin{align}
  N^+(x_i) \cap Y \subseteq N^+(x_j) \cap Y & \iff i \le j \nonumber\\
  \text{and}\quad N^+(y_i) \cap X \subseteq N^+(y_j) \cap X
  & \iff \sigma^{-1}(i) \le \sigma^{-1}(j), \label{eq:half-graph-order}
%
%
%
\end{align}
while in~$\Fc_\le(\sigma)$, the direction of inclusions is reversed.

These encodings of permutations are compatible with patterns.
\begin{lemma}\label{lem:forb-pattern}
  If~$\tau$ is a pattern of~$\sigma$, then~$\Fc_R(\tau)$ is a subtournament of~$\Fc_R(\sigma)$ for each of the relations~$R \in \{=,\le,\ge\}$.
\end{lemma}
\begin{proof}
  If~$\sigma \in \Perm_n$ and~$\tau$ is the pattern induced by $X \subset [n]$,
  then~$\Fc_R(\tau)$ is isomorphic to the subtournament of~$\Fc_R(\sigma)$ induced by the vertices~$x_i$ and~$y_{\sigma(i)}$ for~$i \in X$.
\end{proof}

\begin{lemma}
  For each~$R \in \{=,\le,\ge\}$, the class~$\Fc_R$ efficiently interprets
  the class~$\Oc_\Perm$ of permutations represented as bi-orders.
  Precisely, there is an interpretation~$\Phi_R$, and for any permutation~$\sigma \in \Perm_n$, $n \ge 2$,
  there is a~$\sigma' \in \Perm_{n+1}$ computable in polynomial time such that $\Oc_\sigma = \Phi_R(\Fc_R(\sigma'))$.
  \label{lem:forb-interpret-perm}
\end{lemma}
\begin{proof}
  We will first show that~$\Fc_R(\sigma)$ transduces~$\Oc_\sigma$,
  and then how to remove the colouring step of the transduction by slightly extending~$\sigma$.

  Let~$X = \{x_1,\dots,x_n\}$ and~$Y = \{y_1,\dots,y_n\}$ be as in the definition of~$\Fc_R(\sigma)$.
  The transduction uses colouring to guess the set~$X$.
  It then defines two total orders on~$Y$, which together describe~$\sigma$.
  The first ordering is given by the direction of edges inside~$Y$.
  The second depends on~$R$:
  \begin{itemize}
    \item If~$R$ is~$=$, edges oriented from~$X$ to~$Y$ are a perfect matching.
      The direction of edges in~$X$, interpreted through this matching, defines the second order on~$Y$.
    \item If~$R$ is~$\ge$ or~$\le$, the second order is inclusion, respectively inverse inclusion,
      of in-neighbourhoods intersected with~$X$, see~\eqref{eq:half-graph-order}.
  \end{itemize}
  With the knowledge of which subset is~$X$, each of these orders is clearly definable with a first-order formula.
  Finally, the transduction deletes vertices of~$X$, leaving only~$Y$ and the two orders that encode~$\sigma$.

  Let us now show how to deterministically define the partition~$X,Y$,
  at the cost of extending~$\sigma$ with one fixed value.
  Here, we assume $n \ge 2$.
  \begin{itemize}
    \item If~$R$ is~$=$, define~$\sigma' \in \Perm_{n+1}$
      by $\sigma'(n+1) = n+1$ and $\sigma'(i) = \sigma(i)$ for any $i \le n$.
      Then, in~$\Fc_=(\sigma')$, the unique vertex with out-degree~1 is~$x_{n+1}$.
      Its out-neighbour is~$y_{n+1}$, and we have $X = N^+(y_{n+1}) \cup \{x_{n+1}\}$.
    \item If~$R$ is~$\le$, define $\sigma'(n+1) = 1$ and $\sigma'(i) = \sigma(i)+1$.
      Then~$x_{n+1}$ has out-degree~1, and its out-neighbour is~$y_1$.
      The only other vertex that may have out-degree~$1$ is~$y_{n+1}$, and this happens if and only if $\sigma'(n) = n+1$.
      When this is the case, the direction of the edge~$x_{n+1} \from y_{n+1}$ still allows to distinguish~$x_{n+1}$.
      Finally, $y_1$ is the only out-neighbour of~$x_{n+1}$, and satisfies $X = N^-(y_1)$.
    \item If~$R$ is~$\ge$, define once again~$\sigma'(n+1) = 1$ and $\sigma'(i) = \sigma(i)+1$.
      Then~$y_1$ is the unique vertex with in-degree~$1$, and its in-neighbour is~$x_{n+1}$,
      which satisfies $X = N^-(x_{n+1}) \cup \{x_{n+1}\}$.
  \end{itemize}
  In all three cases, $\Fc_R(\sigma')$ contains two extra vertices compared to~$\Fc_R(\sigma)$.
  These extra vertices can be uniquely identified in first-order logic, and can then be used to define~$X$.
  Combined with the previous transduction, this gives an interpretation of~$\Oc_\sigma$ in~$\Fc_R(\sigma')$.
\end{proof}

We can now prove that the classes~$\Fc_R$ are complex.
\obstrhard*
\begin{proof}
  \Cref{item:independent} is straightforward by \cref{lem:forb-interpret-perm,lem:perm-inter-graph},
  since efficient interpretations can be composed.
  By \cref{lem:interpretation-reduction,thm:FO-aw-hard}, this in turn implies \cref{item:aw-hard}.

  \Cref{item:not-small} implies \cref{item:factorial} by a simple counting argument:
  in an isomorphism class, there are at most~$n!$ choices
  of labelling of vertices with~$\{1,\dots,n\}$ (less when the graph has automorphism).
  Furthermore, each of \cref{item:not-small,item:independent} implies \cref{item:unbounded-tww},
  by~\cref{thm:small} and~\cref{thm:tww-transduction} respectively.
  Thus it only remains to prove \cref{item:not-small}.

  By \cref{lem:forb-interpret-perm}, for all~$\sigma \in \Perm_n$ there is some~$\Fc_R(\sigma')$ on~$2n+2$ vertices
  such that~$\Phi_R(\Fc_R(\sigma')) = \sigma$, where~$\Phi_R$ is an interpretation.
  Since interpretations preserve isomorphism, it follows that there are at least~$n!$
  non-isomorphic tournaments on~$2n+2$ vertices in~$\Fc_R$.
  Furthermore, using the same arguments as in \cref{lem:forb-interpret-perm},
  it is easy to show that these~$\Fc_R(\sigma')$ have no non-trivial automorphism.
  Thus, there are exactly~$(2n+2)!$ distinct labellings of~$\Fc_R(\sigma')$ with~$\{1,\dots,2n+2\}$.
  In total, this gives~$(2n+2)! \cdot n!$ distinct graphs on vertices~$\{1,\dots,2n+2\}$ in~$\Fc_R$,
  proving \cref{item:not-small}.
\end{proof}

Thus the classes~$\Fc_=,\Fc_\le,\Fc_\ge$ are obstructions
to fixed parameter tractability of FO model checking and twin-width.
The rest of the paper shows that they are the only obstructions.
One may also verify that all three are minimal, i.e.\ none of them is contained in another.

\section{Binary search tree orders} \label{sec:bst}
This section presents the good order for twin-width in tournaments.
It is based on \emph{binary search trees} (BST),
which we define in a tournament~$T$ as a rooted ordered binary tree~$S$
(meaning that each node has a left and right child, either of which may be missing),
whose nodes are the vertices of~$T$, and such that for any $x \in S$
\begin{itemize}
  \item the left child of~$x$ (if any) and its descendants are in~$N_T^-(x)$, and
  \item the right child of~$x$ (if any) and its descendants are in~$N_T^+(x)$,
\end{itemize}
see \cref{fig:bst}.

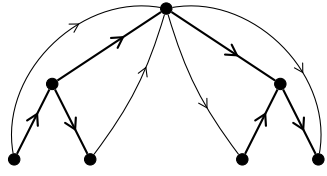
\begin{figure}
  \begin{center}
    \begin{tikzpicture}
      \tikzstyle{every node}=[normalnode]
      \node (r) at (0,0) {};
      \node (0) at (-1.5,-1) {};
      \node (1) at (1.5,-1) {};
      \node (00) at (-2,-2) {};
      \node (01) at (-1,-2) {};
      \node (10) at (1,-2) {};
      \node (11) at (2,-2) {};

      \draw[thick] (00) edge[directed] (0)
                   (0)  edge[directed] (01)
                   (10) edge[directed] (1)
                   (1)  edge[directed] (11)
                   (0)  edge[directed] (r)
                   (r)  edge[directed] (1);

      \draw (01) edge[bend right=10, directed] (r)
            (r)  edge[bend right=10, directed] (10)
            (00) edge[bend left=55, directed] (r)
            (r)  edge[bend left=55, directed] (11);
    \end{tikzpicture}
  \end{center}
  \caption{A binary search tree in a tournament. The direction of omitted edges is not constrained.}
  \label{fig:bst}
\end{figure}

The \emph{order} associated to~$S$, denoted~$<_S$, is the left-to-right order (or in-order traversal) placing a node~$x$ after its left child and its descendants, but before its right child and its descendants.
Such an order is called a \emph{BST order}.

Remark that because~$T$ is only a tournament and not an order as in a standard BST,
there is no restriction on the direction of edges between the left and right subtrees of~$x$.
On the other hand, if~$x$ is an ancestor of~$y$, then
there is an edge oriented from~$x$ to~$y$ if and only if~$x <_S y$.
Thus we have
\begin{remark}
  In a tournament~$T$, any branch~$B$ of a BST~$S$ forms a chain which coincides with~$<_S$.
  That is, for $x,y \in B$, the edge in~$T$ is oriented from~$x$ to~$y$ if and only if $x <_S y$.
  \label{rmk:branch-chain}
\end{remark}

We will now define \emph{chain quasi-orders},
which are FO definable quasi-orders with which we will approximate BST orders.
Let~$C$ be a chain in~$T$.
Its \emph{chain quasi-order}~$\qle_C^+$ is defined as follows.
Enumerate the vertices of~$C$ as~$c_1, \dots, c_k$ so that edges are oriented from~$c_i$ to~$c_j$ when~$i < j$.
Define $A_i = \bigcap_{j \le i} N^+(c_j)$ (where as an edge case~$A_0$ contains all vertices), and $B_i = A_{i-1} \cap N^-(c_i)$.
In other words, $A_{i-1}$ is partitioned as $A_i \uplus \{c_i\} \uplus B_i$,
where~$A_i$ and~$B_i$ contain the out- and in-neighbours of~$c_i$ respectively.
Note that~$B_1,\dots,B_k$ and~$A_k$ are all disjoint from~$C$.

Then each of $B_1,\dots,B_k$, $A_k$, and $\{c_1\},\dots,\{c_k\}$ is an equivalence class of~$\qle_C^+$, and the classes are ordered as
\[ B_1 \qlt_C^+ c_1 \qlt_C^+ B_2 \qlt_C^+ c_2 \qlt_C^+ \dots B_k \qlt_C^+ c_k \qlt_C^+ A_k, \]
see \cref{fig:chain-order}.
This can be seen as the left-to-right order of a partial BST
consisting only of a single branch $c_1,\dots,c_k$, with~$c_1$ as root and~$c_k$ as leaf.
It is also a coarsening of the lexicographic order:
the latter would refine the order inside each class~$B_i$ using~$c_{i+1},\dots,c_k$.

The dual quasi-order~$\qle_C^-$ is defined in the same way, but reversing the direction of all edges.
Thus, we now enumerate~$C$ so that edges are from~$c_i$ to~$c_j$ when~$i > j$,
while~$A_i = \bigcap_{j \le i} N^-(c_i)$ and $B_i = A_{i-1} \cap N^+(c_i)$.
The rest of the definition is the same.

\begin{figure}
  \begin{center}
    \begin{tikzpicture}[>=angle 60]
      \def\n{4}

      \foreach \i in {1,...,\n}{
        \node (c\i) at (1.6*\i,0) [normalnode, label=above:$c_\i$] {};
        \node (B\i) at (1.6*\i-0.7,-1) [inner sep=0pt, minimum size=0pt] {};
        \node (Bl\i) at (1.6*\i-0.7,-1.5) {$B_\i$};

        \draw[directed] (B\i) -- (c\i);
        \draw (B\i) -- +(-0.4,-0.7) -- +(0.4,-0.7) -- +(0,0);
        \foreach \j in {1,...,\n}{
          \ifnum \j<\i
            \draw[directed] (c\j) -- (B\i);
          \fi
        }
      }
      \node (A) at (1.6*\n+1,-1)  [inner sep=0pt, minimum size=0pt] {};
      \node (Al) at (1.6*\n+1,-1.5) {$A_\n$};
      \draw (A) -- +(-0.4,-0.7) -- +(0.4,-0.7) -- +(0,0);
      \foreach \j in {1,...,\n}{
        \draw[directed] (c\j) -- (A);
      }

      \draw[thick,->] (c1) -- (c\n);
    \end{tikzpicture}
  \end{center}
  \caption{%
    Example of construction of the quasi-order~$\qle_C^+$.
    The quasi-order is from left to right, and the triangles are equivalence classes.
    The direction of omitted edges (from~$B_i$ to~$B_j \cup \{c_j\}$ for~$i < j$) is not constrained.
    For~$\qle_C^-$, the direction of all edges would be reversed.
  }
  \label{fig:chain-order}
\end{figure}
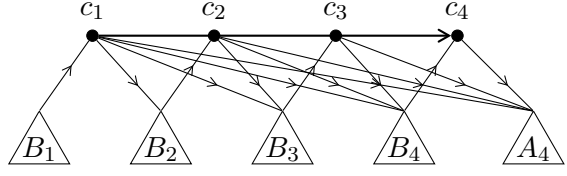

\begin{lemma}
  There is a transduction~$\Phi$ such that for any tournament~$T$ and chain quasi-ordering~$\qle$,
  we have $(T,\qle) \in \Phi(T)$.
  \label{lem:chain-order-transduction}
\end{lemma}
\begin{proof}
  The transduction~$\Phi$ first uses a colouring step to guess a chain~$C$,
  and an orientation~$o \in \{+,-\}$.
  Depending on this non-deterministic choice, it may yield \emph{any} chain quasi-ordering.
  The remainder of~$\Phi$ is a deterministic interpretation.
  For simplicity, we will assume that the orientation~$o$ is~$+$, the case of~$-$ being similar.

  Having guessed~$C$, the ordering~$\qle_C^+$ inside~$C$ is simply given by the orientation of edges.
  Enumerate~$C$ as $\{c_1 \qle_C^+ \dots \qle_C^+ c_k \}$.
  Given~$x \not\in C$, let $i(x) \in \{0,\dots,k\}$ be maximal such that
  there are edges oriented~$c_j \to x$ for all~$1 \le j \le i(x)$.
  Then
  \begin{itemize}[noitemsep]
  \item for any~$x \not\in C$, we have $c_1 \dots c_{i(x)} \qle_C^+ x \qle_C^+ c_{i(x)+1}, \dots, c_k$, and
  \item for~$x,y \not\in C$, we have~$x \qle_C^+ y$ if and only if~$i(x) \le i(y)$.
  \end{itemize}
  This characterisation of~$\qle_C^+$ can be expressed by a first-order formula.
\end{proof}

The main result of this section allows in a sense to approximate BST orderings with chain quasi-orders:
\begin{lemma}
  \label{lem:bst-partition-extraction}
  There is a function~$f(k) = 2^{O(k)}$ such that
  for any tournament~$T$, BST ordering~$<_S$, and family~$\Pc$ of at least~$f(k)$ disjoint intervals of~$<_S$,
  one can find a chain quasi-ordering~$\qle$
  and a subfamily~$\Pc' \subset \Pc$ of at least~$k$ parts
  such that any parts~$X,Y \in \Pc'$ satisfy~$X \qlt Y$ or~$Y \qlt X$.

  Furthermore, $\Pc'$ as well as the chain~$C$ and orientation~$o \in \{+,-\}$ defining~$\qle$ can be computed in linear time.
\end{lemma}
\begin{proof}
  Let~$T$ be a tournament, $S$ a BST of~$T$ and~$<_S$ the corresponding order.
  Consider a family~$\Pc$ of at least~$f(k)$ disjoint intervals of~$<_S$,
  where~$f(k) = 2^{O(k)}$ will be determined later.

  We choose a branch~$B = b_0,\dots,b_p$ of~$S$ by the following process. First~$b_0$ is the root of~$S$.
  For each (yet to be determined)~$b_i$, let~$S_i$ be the subtree of~$S$ rooted in~$b_i$,
  and define the weight~$w_i$ to be the number of intervals of~$\Pc$ intersected by~$S_i$.
  Then~$b_{i+1}$ is chosen to be the child of~$b_i$ that maximizes~$w_{i+1}$.
  This choice ensures that
  \begin{equation}
    \label{eq:branch-weights}
    2w_{i+1} + 1 \ge w_i.
  \end{equation}

  For each~$i < p$, let~$d_i$ be the child of~$b_i$ other than~$b_{i+1}$ (sometimes~$d_i$ does not exist),
  and let~$D_i$ be the subtree of~$S$ rooted at~$d_i$ ($D_i$ is empty if~$d_i$ does not exist).
  Furthermore, let~$L,R$ be the sets of vertices that are before, resp.\ after the leaf~$b_p$ in the order~$<_S$.
  For any $0 \le i \le j \le p$, let
  \[ L_{i,j} \eqdef \bigcup_{\substack{i \le \ell < j \\ b_\ell \in L}} \{b_\ell\} \cup D_\ell,
     \qquad \text{and} \qquad
     R_{i,j} \eqdef \bigcup_{\substack{i \le \ell < j \\ b_\ell \in R}} \{b_\ell\} \cup D_\ell.
  \]
  Roughly speaking, $L_{i,j}$, resp.~$R_{i,j}$ consists of subtrees branching out of~$B$ on the left, resp.\ right, between~$b_i$ and~$b_j$.
  \begin{claim}
    \label{clm:LSR-intervals}
    For any~$i,j$, the subtree~$S_i$ is partitioned into~$L_{i,j} <_S S_j <_S R_{i,j}$.
  \end{claim}
  \begin{proofofclaim}
    Clearly~$L_{i,j},S_j,R_{i,j}$ partition~$S_i$.
    Furthermore, if~$\ell < j$ and~$b_\ell \in L$, then~$b_\ell <_S S_j$, and in turn~$D_\ell <_S b_\ell$.
    This proves~$L_{i,j} <_S S_j$, and symmetrically~$S_j <_S R_{i,j}$.
  \end{proofofclaim}

  \begin{claim}
    For $0 \le i < j \le p$, if $w_i \ge w_j + 3$, then there is a part $P \in \Pc$
    such that $P \subset L_{i,j}$ or $P \subset R_{i,j}$.
    \label{clm:interval-contains-part}
  \end{claim}
  \begin{proofofclaim}
    There are at least three parts of~$\Pc$ that intersect~$S_i$ but not~$S_j$.
    Since these three parts and~$S_i$ are all intervals of~$<_S$, one of these parts, say~$P$, is contained in~$S_i$.
    Thus~$P$ is a subset of~$S_i$ but does not intersect~$S_j$,
    which using \cref{clm:LSR-intervals} implies that $P \subset L_{i,j}$ or $P \subset R_{i,j}$.
  \end{proofofclaim}

  Construct a sequence $i_0 < \dots < i_{2k}$ of indices in~$\{0,\dots,p\}$ inductively by
  taking $i_0 = 0$, and choosing~$i_{\ell+1}$ minimal such that $w_{i_{\ell+1}} \le w_{i_\ell} - 3$.
  Using~\eqref{eq:branch-weights} and the minimality of~$i_{\ell+1}$, we obtain for all~$\ell$ that
  \begin{align}
    2w_{i_{\ell+1}} + 1 & \ge w_{i_{\ell+1} - 1} > w_{i_l} - 3, \\
    \label{eq:subbranch-weights}
    \text{hence} \quad 2w_{i_{\ell+1}} + 3 & \ge w_{i_l}.
  \end{align}
  We can now define~$f$ by~$f(0) = 1$ and~$f(k+1) = 4f(k) + 9$.
  By assumption, $w_0 = \card{\Pc} \ge f(k)$, and it follows from~\eqref{eq:subbranch-weights}
  that the construction of~$i_\ell$ can be carried out up to~$i_{2k}$.

  Define $L'_\ell = L_{i_{\ell-1},i_\ell}$, and similarly $R'_\ell = R_{i_{\ell-1},i_\ell}$, see \cref{fig:bst-branch}.
  By \cref{clm:interval-contains-part}, for any $\ell \in [2k]$,
  either~$L'_\ell$ or~$R'_\ell$ contains a part of~$\Pc$.
  Thus, either there are at least~$k$ distinct~$L'_\ell$ containing a part of~$\Pc$,
  or there are at least~$k$ distinct~$R'_\ell$ containing a part of~$\Pc$.
  Without loss of generality, assume that we are in the former case.
  We will now forget about vertices which are not in~$L$.

  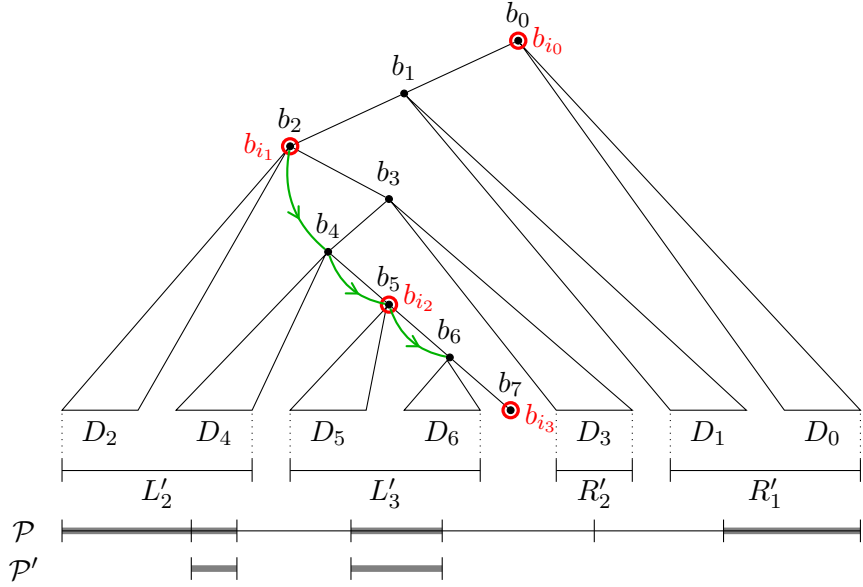
\begin{figure}[tb]
    \begin{center}
      \begin{tikzpicture}
        \def\s{0.7}
        \def\bb{-7*\s}

        \foreach \d/\x in {0/1,1/-0.5,2/-2,3/-0.7,4/-1.5,5/-0.7,6/0.1,7/0.9}{
          \node (b\d) at (\x,-\d*\s) [smallnode, label=above:$b_\d$] {};
        }
        \def\b{\bb}

        \foreach \d/\l/\r in {0/4.5/5.5,1/3/4,3/1.5/2.5, 6/-0.5/0.5,5/-2/-1,4/-3.5/-2.5,2/-5/-4}{
          \draw (b\d) -- (\l,\b) -- (\r,\b) -- cycle;
          \node (D\d) at (\l+0.5,\b-0.3) {$D_\d$};
        }
        \draw (b0) -- (b1) -- (b2) -- (b3) -- (b4) -- (b5) -- (b6) -- (b7);

        \foreach \i in {0,2,5,7}{
          \draw[blue,very thick] (b\i) circle (1mm);
        }
        \draw (b0) +(0.4,0) node[blue] {$b_{i_0}$};
        \draw (b2) +(-0.4,0) node[blue] {$b_{i_1}$};
        \draw (b5) +(0.4,0.1) node[blue] {$b_{i_2}$};
        \draw (b7) +(0.4,-0.1) node[blue] {$b_{i_3}$};

        \foreach \i/\j in {2/4,4/5,5/6}{
          \draw (b\i) edge [bend right=30,green!70!black,thick,directed] (b\j);
        }

        \def\b{\bb - 0.8}
        \foreach \r/\l/\c/\L in {5.5/3/4.25/$R'_1$,2.5/1.5/2/$R'_2$,0.5/-2/-0.75/$L'_3$,-2.5/-5/-3.75/$L'_2$}{
          \draw (\l,\b-0.15) -- +(0,0.3);
          \draw (\r,\b-0.15) -- +(0,0.3);
          \draw[dotted] (\l,\b+0.15) -- (\l,\bb);
          \draw[dotted] (\r,\b+0.15) -- (\r,\bb);
          \draw (\l,\b) -- (\r,\b);
          \draw (\c,\b-0.3) node {\L};
        }

        \def\b{\bb - 1.6}
        \node (P) at (-5.5,\b) {$\Pc$};
        \draw (-5,\b) -- (5.5,\b);
        \foreach \l/\r in {3.7/5.5,-1.2/0,-2.7/-5}{
          \draw[line width=2.5pt, draw opacity=0.5] (\l,\b) -- (\r,\b);
        }
        \foreach \x in {5.5,3.7,2,0,-1.2,-2.7,-3.3,-5}{
          \draw (\x,\b-0.15) -- +(0,0.3);
        }

        \def\b{\bb - 2.1}
        \node (PP) at (-5.5,\b) {$\Pc'$};
        \foreach \r/\l in {-3.3/-2.7,-1.2/0}{
          \draw (\l,\b-0.15) -- +(0,0.3);
          \draw (\r,\b-0.15) -- +(0,0.3);
          \draw[line width=2.5pt, opacity=0.5] (\l,\b) -- (\r,\b);
        }
      \end{tikzpicture}
    \end{center}
    \caption{%
      Sketch of the proof of \cref{lem:bst-partition-extraction}.
      In the upper half, the BST~$T$ with the extracted branch~$B$;
      circled in blue, the extracted subsequence~$b_{i_\ell}$;
      in green arrows, the chain~$C = B \cap L = \{b_2,b_4,b_5,b_6\}$.
      Below the tree, from top to bottom:
      the partition in~$L'_\ell$ and~$R'_\ell$;
      the initial family (here partition)~$\Pc$, with the parts contained in some~$L'_\ell$ or~$R'_\ell$ highlighted;
      the final family~$\Pc'$, obtained by selecting a part of~$\Pc$ inside each possible~$L'_\ell$.
    }
    \label{fig:bst-branch}
  \end{figure}

  Define~$C = L \cap B$. By \cref{rmk:branch-chain}, this is a chain, whose order coincides with~$<_S$.
  Furthermore, at any node~$x$ of~$C$, the branch~$B$ descends on the right side, since~$x <_S b_p$.
  Thus, the order in~$C$ also coincides with the ancestor-descendent order of~$S$.
  (Remark here that if we were in~$R$ instead of~$L$,
  the order of~$C$ would be the inverse of the ancestor-descendant order.)
  Now, if~$C$ is enumerated as~$c_0 <_S \dots <_S c_t$,
  and~$C_i$ is the subtree branching out on the left of~$c_i$, defined similarly to~$D_i$,
  then the chain quasi-order~$\qle_C^+$ restricted to~$L$ is exactly
  \[ C_0 \qlt_C^+ c_0 \qlt_C^+ C_1 \qlt_C^+ c_1 \qlt_C^+ \dots \qlt_C^+ c_t \]
  where each subtree~$C_i$ is an equivalence class.
  (In~$R$, we would instead use~$\qle_C^-$.)
  From this description, we obtain that any~$L_{i,j}$ is an interval of~$\qle_C^+$ restricted to~$L$.

  For each~$L'_\ell$, select a part of~$\Pc$ included in~$L'_\ell$ if any,
  and define~$\Pc'$ as the collection of selected parts.
  Thus~$\Pc' \subset \Pc$, and we know from the choice of the family~$\{L'_\ell\}_{\ell \in [2k]}$
  that $\card{\Pc'} \ge k$.
  Furthermore, if~$X \neq Y$ are parts of~$\Pc'$, there are $s \neq t$
  such that~$X \subseteq L'_s$ and~$Y \subseteq L'_t$.
  Since the~$L'_\ell$ are disjoint intervals of~$(L,\qle_C^+)$,
  we have either~$L'_s \qlt_C^+ L'_t$ or~$L'_t \qlt_C^+ L'_S$,
  hence a fortiori either~$X \qlt_C^+ Y$ or~$Y \qlt_C^+ X$.
  Thus~$\Pc'$ satisfies all desired properties.

  Finally, given the BST~$S$ and the family~$\Pc$,
  it is routine to compute the weights~$w_i$ of all nodes in~$S$ by a bottom-up procedure;
  this only requires to compute the left-most and right-most parts of~$\Pc$ intersecting each subtree.
  From this, it is simple to choose in linear time
  the branch~$B$, the indices~$i_\ell$, the better side~$L$ or~$R$, and finally to compute~$C$ and~$\Pc'$.
\end{proof}

\section{Approximating twin-width}\label{sec:approx}
In this section, we prove \cref{lem:BST-twinwidth}: for any tournament~$T$ and BST ordering~$<$,
if~$T$ has small twin-width, then so does the ordered structure~$(T,<)$.
This implies an \fpt approximation of twin-width in tournaments (\cref{thm:approx-algo}).

To this end, we show using \cref{thm:grid-theorem,lem:bst-partition-extraction} that if~$(T,<)$ has large twin-width,
then~$T$ contains a so-called \emph{chain order representation} of a matrix with high grid rank.
It is relatively simple to show that these chain order representations are obstructions to twin-width.
In \cref{sec:extraction}, we will improve this result
by extracting the canonical obstructions~$\Fc_=,\Fc_\le,\Fc_\ge$ from chain order representations.

\subsection{Chain order representation}\label{sec:chain-representation}
Let~$M$ be a $0,1$-matrix and~$T$ a tournament.
A \emph{chain order representation} of~$M$ in~$T$ consists of
two disjoint subsets~$A,B$ of vertices and two chain quasi-orders~$\qle^A,\qle^B$ such that
\begin{itemize}[noitemsep]
  \item restricted to~$A$ and~$B$ respectively, $\qle^A$ and~$\qle^B$ are strict orderings,
  \item and the adjacency matrix of~$(A,\qle^A)$ versus~$(B,\qle^B)$ is~$M$.
\end{itemize}
See \cref{fig:chain-representation} for an example.
Explicitly, the matrix~$M$ has columns~$(A,\qle^A)$, rows~$(B,\qle^B)$,
and contains a~`1' at position~$(a,b)$ for~$a \in A$ and~$b \in B$ if and only if~$a \to b$ is an edge.
The subsets~$A,B$ should not be confused with the chains~$C_A,C_B$ from which the orders~$\qle^A,\qle^B$ are defined.
Note also that only~$A$ and~$B$ are required to be disjoint; the sets~$A,B,C_A,C_B$ may otherwise intersect.

\begin{figure}
  \begin{center}
    \begin{tikzpicture}
      \begin{scope}[yscale=0.7,xscale=0.9]
      \foreach \i in {1,...,5}{
        \node (c\i) [smallnode] at (0,\i) {};
        \node (a\i) [smallnode] at (2,\i) {};
        \node (b\i) [smallnode] at (4,\i) {};
        \node (d\i) [smallnode] at (6,\i) {};
      }
      \draw[thick,->] (c1) -- (c5);
      \draw[thick,->] (d1) -- (d5);

      \draw (a3.center) ellipse (0.3 and 3);
      \draw (b3.center) ellipse (0.3 and 3);
      \draw (2,-0.5) node {$A$};
      \draw (4,-0.5) node {$B$};
      \draw (0,0.5) node {$C_A$};
      \draw (6,0.5) node {$C_B$};

      \foreach \i/\si in {1/3,2/1,3/4,4/5,5/2}{
        \draw[directed] (a\i) -- (b\si);
      }

      \foreach \i in {1,...,5}{
        \draw[directed] (a\i) -- (c\i);
        \draw[directed] (b\i) -- (d\i);

        \foreach \j in {1,...,5}{
          \ifnum \j<\i
            \draw[directed] (c\j) -- (a\i);
            \draw[directed] (d\j) -- (b\i);
          \fi
        }
      }
      \end{scope}

      \begin{scope}[xshift=7cm,yshift=0.3cm,scale=0.6]
      \draw (0.5,0.5) rectangle (5.5,5.5);

      \foreach \i/\si in {1/3,2/1,3/4,4/5,5/2}{
        \foreach \j in {1,...,5}{
          \ifnum \j=\si
            \node at (\i,\j) {1};
          \else
            \node at (\i,\j) {0};
          \fi
        }
      }

      \node at (3,0) {$M$};
      \end{scope}
    \end{tikzpicture}
  \end{center}
  \caption{%
    Chain representation of the matrix~$M$ of the permutation~$31452$.
    Vertices of~$A$ are ordered bottom to top by a chain order~$\qle^+_{C_A}$, and similarly for~$B$.
    Edges oriented from~$A$ to~$B$ correspond to~`1's in~$M$.
    For readability, edges from~$B$ to~$A$ (corresponding to~`0's) are not drawn.
    The other omitted edges are unconstrained.
  }
  \label{fig:chain-representation}
\end{figure}

\begin{lemma}
  Let~$\Tc$ be a class of tournaments containing
  chain order representations of matrices with rank-$k$ divisions for arbitrarily large~$k$.
  Then~$\Tc$ has unbounded twin-width.
  \label{lem:chain-order-obstruction}
\end{lemma}
\begin{proof}
  Consider a chain representation of~$M$ in~$T$ defined by subsets~$A,B$ and chain orders~$\qle^A,\qle^B$.
  Construct the ordered structure~$(T',<)$
  where~$T' = T[A \cup B]$ is the subtournament induced by~$A$ and~$B$,
  and the ordering~$<$ is obtained by combining~$\qle^A$ and~$\qle^B$:
  $<$ coincides with~$\qle^A$ and~$\qle^B$ inside~$A$ and~$B$ respectively, and~$A < B$.
  Thus~$M$ is contained in the adjacency matrix of~$T'$ ordered by~$<$.
  Call~$\Tc'$ the class of all such ordered structures constructed from~$T \in \Tc$.
  By assumption, structures in~$\Tc'$ have adjacency matrices with arbitrarily high rank divisions,
  hence by \cref{thm:grid-theorem} $\Tc'$ has unbounded twin-width.

  Using \cref{lem:chain-order-transduction}, it is simple to show that
  there is a fixed first-order transduction~$\Phi$ that given a tournament~$T$,
  produces any ordering~$<$ of~$T$ obtained by combining two chain orders as previously.
  It follows that~$\Tc'$ is obtained by first-order transduction from~$\Tc$.
  Since~$\Tc'$ has unbounded twin-width, this implies by \cref{thm:tww-transduction}
  that~$\Tc$ has unbounded twin-width.
\end{proof}

\subsection{Self-extraction in high-rank divisions}
The last tool needed to prove \cref{lem:BST-twinwidth}
is the following self-extraction lemma for high-rank divisions.
\begin{lemma}
  \label{lem:grid-in-grid}
  Let~$M$ be a matrix with a rank-$(k^3)$ division~$(\Rc,\Cc)$.
  Then~$M$ has a submatrix~$M'$ with a rank-$k$ division
  such that~$M'$ contains only one row (resp.\ column) from each part of~$\Rc$ (resp.~$\Cc$).
\end{lemma}
\Cref{lem:grid-in-grid} will be used as follows:
from a high-rank division whose parts are strictly ordered by a chain-quasi order~$\qle$
(while allowing elements of the same part to be equivalent under~$\qle$),
it yields a high-rank division on a subset of elements totally ordered by~$\qle$.

We first extract a high-rank matrix of the desired form from a high-rank division.
\begin{lemma}
  \label{lem:rank-in-grid}
  Let~$M$ be a matrix with a rank-$k$ division~$(\Rc,\Cc)$.
  Then there exists a submatrix~$M'$ with \emph{rank~$k$}
  containing only one row (resp.\ column) from each part of~$\Rc$ (resp.~$\Cc$).
\end{lemma}
\begin{proof}
  Enumerate the blocks $\Rc = \{R_1,\dots,R_k\}$ and $\Cc = \{C_1,\dots,C_k\}$ of the division.
  Assume by induction that we have already extracted one row (resp.\ column)
  from each of~$R_1,\dots,R_{k-1}$ (resp.~$C_1,\dots,C_{k-1}$) to form a submatrix~$M_{k-1}$ of rank~$k-1$.
  Now add all rows of~$R_k$ and all columns of~$C_k$ to~$M_{k-1}$.
  The resulting matrix has rank at least~$k$ since it contains~$R_k \times C_k$.
  Using basis exchange, all but one row of~$R_k$ can be removed while preserving rank~$k$,
  and similarly all but one column of~$C_k$ can be removed, yielding the desired submatrix.
\end{proof}

\begin{proof}[Proof of \cref{lem:grid-in-grid}]
  Let us enumerate the blocks of the division~$(\Rc,\Cc)$ as
  $\Rc = \{R_0 < \dots < R_{k^3-1}\}$ and $\Cc = \{C_0 < \dots < C_{k^3-1}\}$.
  We will use two levels of coarser partitions, regrouping blocks~$k$ by~$k$:
  define the partition~$\Rc'$ with parts $R'_i = R_{ki} \cup \dots \cup R_{k(i+1) - 1}$ for~$i < k^2$,
  as well as~$\Rc''$ with parts $R''_i = R'_{ki} \cup \dots \cup R'_{k(i+1)-1}$ for~$i < k$,
  and similarly~$\Cc'$ and~$\Cc''$.
  Thus~$(\Rc',\Cc')$ is a $k^2$-division of~$M$, while~$(\Rc'',\Cc'')$ is a $k$-division.

  Now for each~$i,j \in [k]$, we consider the submatrix~$R'_{ki + j} \times C'_{kj + i}$.
  Remark that these submatrices all have disjoint sets of rows and columns.
  Further, $R'_{ki + j} \times C'_{kj + i}$ contains~$k$ parts of~$\Rc$ and of~$\Cc$,
  hence \cref{lem:rank-in-grid} can be applied to extract a rank-$k$ submatrix~$M_{i,j}$
  using only one row or column from the relevant blocks of~$\Rc$ and~$\Cc$.
  Define~$M'$ to be the submatrix formed only by the rows and columns used in some~$M_{i,j}$.
  The former construction ensures this is only one row or column from each part of~$\Rc$ and~$\Cc$.
  In~$M'$, $(\Rc'',\Cc'')$ is a $k$-division whose cell~$R''_i \times C''_j$
  contains the rank~$k$ submatrix~$M_{i,j}$.
  Therefore~$M'$ has a rank-$k$ division.
\end{proof}

\subsection{Finding high grid rank representations}
We are now ready to prove that BST orders yield an approximation of twin-width in tournaments.
Given a class~$\Tc$ of tournaments, denote by~$\Tc^{BST}$ the class of ordered tournaments~$(T,<_S)$
with~$T \in \Tc$ and~$<_S$ some BST ordering of~$T$.
\begin{lemma}
  \label{lem:BST-chain-order-repr}
  Let~$\Tc$ be a class of tournaments.  If~$\Tc^{BST}$ has unbounded twin-width,
  then~$\Tc$ contains chain order representations of matrices with rank-$k$ divisions for arbitrarily large~$k$.
\end{lemma}
\begin{proof}
  Assume that~$\Tc^{BST}$ has unbounded twin-width.
  Since it is a class of ordered structures, we can apply the results of~\cite{twin-width4}
  to obtain matrices with large grid rank:
  \begin{claim}
    For any~$k$, there exists~$(T,<) \in \Tc^{BST}$
    and intervals of vertices $A_1 < \dots < A_k < B_1 < \dots < B_k$
    such that for any~$i,j$, the adjacency matrix of~$A_i$ versus~$B_j$ has rank~$k$.
    \label{clm:rank-BST}
  \end{claim}
  \begin{proofofclaim}
    By \cref{thm:grid-theorem}, ordered structures with sufficiently large twin-width
    have rank-$k$ divisions in their adjacency matrix, for any given~$k$.
    For~$\Tc^{BST}$ which has unbounded twin-width,
    this yields~$(T,<) \in \Tc^{BST}$ and a rank-$(2k)$ division
    consisting of two partitions~$A_1 < \dots < A_{2k}$ and $B_1 < \dots < B_{2k}$
    of~$V(T)$ such that the adjacency matrix of~$A_i$ versus~$B_j$ has rank~$k$.

    Let us ensure that the~$A_i$s and~$B_j$s are disjoint.
    If the last element of~$A_k$ is smaller than the first of~$B_{k+1}$,
    then~$A_1 < \dots < A_k < B_{k+1} < \dots < B_{2k}$ satisfy the desired property.
    Otherwise, the last element of~$B_k$ is smaller than the first of~$A_{k+1}$,
    and symmetrically~$B_1 < \dots < B_k < A_{k+1} < \dots < A_{2k}$ are as desired.
  \end{proofofclaim}

  We now wish to replace the BST ordering~$<$ of \cref{clm:rank-BST} with chain quasi-orders.
  \begin{claim}
    For any~$k$, there exists~$T \in \Tc$, two chain quasi-orders~$\qle^A,\qle^B$ in~$T$,
    and subsets $A_1 \qlt^A \dots \qlt^A A_k$ and $B_1 \qlt^B \dots \qlt^B B_k$ of~$V(T)$
    with all~$A_i$s and~$B_j$s pairwise disjoint,
    such that the adjacency matrix of~$A_i$ versus~$B_j$ has rank~$k$.
    \label{clm:rank-chain-order}
  \end{claim}
  \begin{proofofclaim}
    Applying~\cref{clm:rank-BST} with~$f(k)$ where~$f$ is the function of \cref{lem:bst-partition-extraction},
    there exists~$T \in \Tc$ and intervals~$A_1 < \dots < A_{f(k)} < B_1 < \dots < B_{f(k)}$
    of some BST order~$<$ such that the adjacency matrix of~$A_i$ versus~$B_j$ has rank~$k$.
    Applying \cref{lem:bst-partition-extraction} to the~$A_i$s and to the~$B_j$s separately
    yields subfamilies of size at least~$k$, both of which are strictly ordered by some chain quasi-order.
  \end{proofofclaim}

  To obtain the chain order representation from \cref{clm:rank-chain-order},
  we only need to ensure that~$\qle^A,\qle^B$ are also strict orderings inside each part~$A_i$ or~$B_j$.
  This is obtained thanks to \cref{lem:grid-in-grid}.

  Apply \cref{clm:rank-chain-order} with~$k^3$, yielding parts~$A_i,B_j$ ordered by~$\qle^A,\qle^B$.
  Consider the adjacency matrix~$M$ of~$\bigcup_i A_i$ against~$\bigcup_j B_j$
  ordered by~$\qle^A$ and~$\qle^B$ respectively, breaking equivalences in the quasi-orderings arbitrarily.
  Since the matrix of~$A_i$ against~$B_j$ has rank at least~$k^3$, $M$ has a rank-$(k^3)$ division.
  We then apply \cref{lem:grid-in-grid} to extract~$A,B$ with only one element in each~$A_i,B_j$,
  so that the restriction to~$A,B$ has a rank-$k$ division.
  Then the chain quasi-orderings~$\qle_A,\qle_B$ strictly order~$A$ and~$B$.
  Thus~$A,B$ define a chain order representation of a matrix with a rank-$k$ division, as desired.
\end{proof}

From \cref{lem:chain-order-obstruction,lem:BST-chain-order-repr},
we immediately obtain that if~$\Tc^{BST}$ has unbounded twin-width, then so does~$\Tc$, or equivalently:
\bsttww*

Using \cref{lem:BST-twinwidth}, we obtain an \fpt approximation algorithm for twin-width of tournaments.
\apxalgo*
\begin{proof}
  Given a tournament~$T$, construct an arbitrary BST ordering~$<$.
  By \cref{lem:BST-twinwidth}, $\tww(T,<)$ is bounded by some function of~$\tww(T)$.
  We then apply to the ordered structure~$(T,<)$
  the approximation algorithm of ordered structures~\cref{thm:grid-theorem}, which is \fpt.
  It finds a contraction sequence for~$(T,<)$
  of width bounded by a function of~$\tww(T,<)$, hence by a function of~$\tww(T)$.
  The same contraction sequence is also valid for~$T$, and forgetting the ordering~$<$ does not increase its width.
\end{proof}

\section{Extracting forbidden tournaments} \label{sec:extraction}
The goal of this section, the last part of the proof of our main result,
is to show that any class of tournaments with unbounded twin-width
contains one of the classes of obstructions described in \cref{sec:obstructions}.

Recall that~$\Tc^{BST}$ denotes the class of ordered tournaments~$(T,<_S)$
where $T \in \Tc$ and~$<_S$ is some BST ordering of~$T$.
The precise statement proved in this section is the following.
\begin{theorem}[\cref{thm:forbidden-tournaments-intro} restated]
  Let~$\Tc$ be a hereditary class of tournaments.
  If~$\Tc^{BST}$ has unbounded twin-width, then~$\Tc$ contains one of the classes~$\Fc_=,\Fc_\le,\Fc_\ge$.
  \label{thm:forbidden-tournaments}
\end{theorem}
Since the classes~$\Fc_=,\Fc_\le,\Fc_\ge$ have unbounded twin-width,
this strengthens \cref{lem:BST-twinwidth}.
The proof builds upon that of \cref{lem:BST-chain-order-repr} in the previous section.
The latter constructed chain order representations of matrices with high-rank divisions.
From these, we will extract encodings of arbitrary permutations
thanks to a major technical result of \cite{twin-width4},
and further refine these encodings to obtain the classes~$\Fc_=,\Fc_\le,\Fc_\ge$ using Ramsey arguments.

\subsection{Encoding permutations in matrices}\label{sec:perm-matrice-encoding}
\newcommand{\matenc}{\mathcal{E}}

Let us described the encodings of permutations in matrices
used as obstructions to twin-width in~\cite{twin-width4}.
They define six classes of matrices $\Mc_=,\Mc_{\neq},\Mc_{\le_R},\Mc_{\ge_R},\Mc_{\le_C},\Mc_{\ge_C}$.\footnote{%
  The classes are called~$\Fc_=$, etc.\ in~\cite{twin-width4},
  we use~$\Mc_=$ to avoid ambiguity with the tournaments classes defined in \cref{sec:obstructions}.
}
Like the classes of tournaments described in \cref{sec:obstructions},
each class~$\Mc_s$ consists of the hereditary closure
of encodings~$\Mc_s(\sigma)$ of arbitrary permutations~$\sigma$.
For a permutation~$\sigma$ on~$n$ elements, the $n \times n$ matrix
$\Mc_s(\sigma)$ is defined as follows, where~$m_{i,j}$ denotes the entry at position~$(i,j)$:
\begin{align*}
  \Mc_=(\sigma):       && m_{i,j} = 1 & \iff j = \sigma(i) \\
  \Mc_{\neq}(\sigma):  && m_{i,j} = 1 & \iff j \neq \sigma(i) \\
  \Mc_{\le_R}(\sigma): && m_{i,j} = 1 & \iff j \le \sigma(i) \\
  \Mc_{\ge_R}(\sigma): && m_{i,j} = 1 & \iff j \ge \sigma(i) \\
  \Mc_{\le_C}(\sigma): && m_{i,j} = 1 & \iff i \le \sigma^{-1}(j) \\
  \Mc_{\ge_C}(\sigma): && m_{i,j} = 1 & \iff i \ge \sigma^{-1}(j)
\end{align*}
Thus~$\Mc_=(\sigma)$ is the usual permutation matrix, $\Mc_{\neq}(\sigma)$ is its complement,
and $\Mc_{\le_R},\Mc_{\ge_R},\Mc_{\le_C},\Mc_{\ge_C}$ are obtained from the permutation matrix
by propagating the `1's down, up, left, and right respectively, see \cref{fig:matrix-encodings}.
Observe that if~$\tau$ is a pattern of~$\sigma$, then~$\Mc_s(\tau)$ is a submatrix of~$\Mc_s(\sigma)$ for each of these encodings~$s$.
\begin{figure}[htb]
  \begin{center}
    \begin{tikzpicture}[scale=0.2]

      \foreach \tst/\x in {
        {\j == \si}/0,{\j != \si}/1,{\j <= \si}/2,{\j >= \si}/3,{\i <= \sj}/4,{\i >= \sj}/5
      }{
        \begin{scope}[xshift=\x{}0cm]
          \foreach \i in {0,...,8}{
            \foreach \j in {0,...,8}{
              \pgfmathtruncatemacro{\iq}{floor(\i/3)}
              \pgfmathtruncatemacro{\ir}{mod(\i,3)}
              \pgfmathtruncatemacro{\si}{\ir*3+2-\iq}
              \pgfmathtruncatemacro{\jq}{floor(\j/3)}
              \pgfmathtruncatemacro{\jr}{mod(\j,3)}
              \pgfmathtruncatemacro{\sj}{6-\jr*3+\jq}
              \pgfmathsetmacro{\c}{\tst ? "black" : "white"}
              \fill[\c] (\i,\j) rectangle +(1,1);
            }
          }
          \draw[white!50!black] (0,0) grid (9,9);
        \end{scope}
      }

    \end{tikzpicture}
  \end{center}
  \caption{%
    The six encodings $=,\neq,\le_R,\ge_R,\le_C,\ge_C$ (left to right) of the 3-grid permutation $369258147$.
    The~1 and~0 are drawn as black and white squares respectively.
    (Figure from~\cite{twin-width4}.)
  }
  \label{fig:matrix-encodings}
\end{figure}

We denote by $\matenc = \{=,\neq,\le_R,\ge_R,\le_C,\ge_C\}$ the symbols denoting these six encodings.
Thus, for~$s \in \matenc$, $\Mc_s$ is one of the previously described classes of matrices.
The main technical result of~\cite{twin-width4} is the following.
\begin{theorem}[{\cite[Theorem~29]{twin-width4}}]
  \label{thm:rank-division-extraction}
  There is a function~$f$ such that in any matrix~$M$ with a rank-$f(k)$ division,
  there is an encoding $s \in \matenc$ such that
  $\Mc_s(\sigma)$ is a submatrix of~$M$ for any permutation~$\sigma$ on~$k$ elements.
\end{theorem}

Consider now a class~$\Tc$ such that~$\Tc^{BST}$ has unbounded twin-width.
By \cref{lem:BST-chain-order-repr}, $\Tc$ contains chain order representations
of matrices with arbitrarily high rank divisions.
Applying \cref{thm:rank-division-extraction} to these matrices, we obtain the following.
\begin{proposition}
  For any permutation $\sigma$, there exists an encoding $s \in \matenc$
  such that some~$T \in \Tc$ contains a chain order representation of~$\Mc_s(\sigma)$.
  \label{clm:chain-permutation1}
\end{proposition}

\subsection{Permutations between chains}
\Cref{clm:chain-permutation1} proves that there are chain order representations of arbitrary permutations.
Such a representation involves four sets of vertices: the disjoint subsets~$A,B$ inducing~$\Mc_s(\sigma)$ as adjacency matrix,
and the two chains from which the orderings of~$A$ and~$B$ are defined.
The next result shows that only two of those four sets are needed to encode the permutations.

\begin{proposition}
  \label{clm:chain-permutation2}
  For any permutation~$\sigma \in \Perm_n$, there exists in some~$T \in \Tc$ one of the following two structures:
  \begin{enumerate}
    \item \label{item:center-matrix} Either two disjoint ordered subsets
      $A = \{a_1 <_A \dots <_A a_n \}$, and $B = \{b_1 <_B \dots <_B b_n \}$ of vertices such that
      \begin{itemize}[nosep]
        \item $A$ is a chain corresponding to the order~$<_A$ up to reversal,
          meaning that either for all~$i < j$ there is an edge~$a_i \to a_j$,
          or for all~$i < j$ there is an edge~$a_i \from a_j$,
        \item similarly~$B$ is a chain corresponding to~$<_B$ up to reversal, and
        \item the adjacency matrix of~$(A,<_A)$ versus~$(B,<_B)$
          is~$\Mc_s(\sigma)$ for some encoding $s \in \matenc$.
      \end{itemize}
    \item \label{item:side-chain-order} Or a subset $A = \{a_1 <_A \dots <_A a_n\}$ and a chain quasi-order~$\qlt$ such that
      \begin{itemize}[nosep]
        \item $\qlt$ strictly orders $A$,
        \item $A$ is a chain corresponding to~$<_A$ up to reversal, and
        \item the bi-order $(A,\qlt,<_A)$ is isomorphic to either~$\Oc_\sigma$ or~$\Oc_{\sigma^{-1}}$.
      \end{itemize}
  \end{enumerate}
\end{proposition}
\begin{proof}
  Let~$\tau$ be a permutation, which will be chosen as a function of~$\sigma$ so as to apply \cref{lem:perm-ramsey}.
  By \cref{clm:chain-permutation1}, there exists a chain order representation of~$\tau$,
  consisting of disjoint subsets~$A,B$ ordered by~$\qle_1,\qle_2$ respectively,
  such that the adjacency matrix of~$A$ versus~$B$ for these orders
  is~$\Mc_s(\tau)$ for some encoding~$s \in \matenc$.
  Enumerate $A = \{a_1 \qle_1 \dots \qle_1 a_n\}$ and $B = \{b_1 \qle_2 \dots \qle_2 b_n\}$.

  Consider the bijection $\phi : A \to B$ defined as $\phi(a_i) = b_{\tau(i)}$, which is in a sense represented by~$\Mc_s(\tau)$.
  Through this bijection, we can transfer the orderings~$\qle_1$ to~$B$ and~$\qle_2$ to~$A$.
  That is, we define the ordering~$\qle'_2$ of~$A$ and~$\qle'_1$ of~$B$ by
  \begin{align*}
    a_i \qlt'_2 a_j \quad & \iff \quad \phi(a_i) \qle_2 \phi(a_j) \quad \iff \quad \tau(i) < \tau(j) \\
    \text{and} \qquad b_i \qlt'_1 b_j \quad & \iff \quad \phi^{-1}(b_i) \qle_1 \phi^{-1}(b_j) \quad \iff \quad \tau^{-1}(i) < \tau^{-1}(j).
  \end{align*}

  Recall that the biorder~$\Oc_\tau = ([n],<,<_\tau)$ consists of the natural ordering~$<$,
  and the permutated ordering defined by~$i <_\tau j$ if and only if~$\tau(i) < \tau(j)$.
  Observe that~$\Oc_\tau$ is isomorphic to both~$(A,\qle_1,\qle'_2)$ and~$(B,\qle'_1,\qle_2)$,
  through the maps $i \mapsto a_i$ and $i \mapsto b_{\tau(i)}$ respectively.
  On this structure, we define a colouring~$\lambda : [n]^2 \to \{0,1\}^2$
  that combines the adjacency matrices of~$T[A]$ and~$T[B]$, that is
  \[ \lambda(i,j) \eqdef (a,b) \quad \text{where} \quad
    \begin{cases}
      a=1 & \iff \text{ $a_i \to a_j$ is an edge} \\
      b=1 & \iff \text{ $b_{\tau(i)} \to b_{\tau(j)}$ is an edge}
    \end{cases}
  \]

  We now suppose that~$\tau$ was chosen so that \cref{lem:perm-ramsey} can be applied
  to obtain a subset~$Q \subset [n]$ such that~$\Oc_\tau[Q]$ is isomorphic to~$\Oc_\sigma$,
  and the colouring~$\lambda$ restricted to~$Q \times Q$ depends only on the orderings~$<,<_\tau$.

  Consider $A' \eqdef \setst{a_i}{i \in Q}$ and $B' \eqdef \setst{b_{\tau(i)}}{i \in Q}$ be the corresponding subsets of~$A,B$.
  They form a chain representation of~$\Mc_s(\sigma)$.
  Furthermore, the direction of edges in~$T[A']$, respectively~$T[B']$,
  depends only on the orders~$\qle_1,\qle'_2$, respectively~$\qle'_1,\qle_2$.
  By \cref{lem:tournament-order}, this means that~$T[A']$ is a transitive tournament
  whose order corresponds to---in the sense of equal to or inverse of---%
  either~$\qle_1$ or~$\qle'_2$, and similarly in~$B'$.
  We now have to consider two cases.
  \begin{enumerate}
    \item If the edges in~$T[A']$ correspond to the ordering~$\qle_1$,
      and the edges in~$T[B']$ correspond to the order~$\qle_2$, then the chain orders become redundant:
      we have obtained~$\Mc_s(\sigma)$ as adjacency matrix between two chains, up to some reversal of orderings.
      This is case~\labelcref{item:center-matrix} of the statement.
    \item Otherwise, if for instance the edges in~$T[A']$ correspond to the order~$\qle'_2$,
      then~$A'$ is equipped with the orderings~$\qle_1$ (as chain ordering) and~$\qle'_2$ (from the direction of edges),
      which together encode~$\sigma$ as a bi-order.
      This is case~\labelcref{item:side-chain-order} of the statement.

      The situation is similar when the edges of~$T[B']$ correspond to~$\qle'_1$,
      except the roles of the chain ordering and the ordering induced by edges are swapped,
      which corresponds to replacing~$\sigma$ by its inverse. \qedhere
  \end{enumerate}
\end{proof}

\subsection{Merging the two cases}
Next, we reduce case~\ref{item:side-chain-order} of \cref{clm:chain-permutation2} to case~\ref{item:center-matrix} by using \cref{thm:rank-division-extraction} a second time.
To this end, we need to obtain high rank divisions from chain quasi-orderings.

\begin{lemma}\label{lem:rank-minor-chain-order}
  In a tournament~$T$, consider a chain quasi-ordering~$\qlt$ defined by a chain~$C$,
  and an ordered subset of vertices~$(X,<_X)$ disjoint from~$C$
  such that~$\qlt$ is a total ordering on~$X$, and the biorder $(X,\qlt,<_X)$ describes a permutation with a $(k^2+k+1)$-grid.
  Then the adjacency matrix of~$(C,\qlt)$ versus~$(X,<_X)$ has a rank-$k$ division.
\end{lemma}
\begin{proof}
  Consider the chain~$C$ and orientation~$o$ that define the chain ordering~$\qlt$, i.e.~$\qlt$ is~$\qlt_C^o$.
  We assume that the orientation~$o$ is~`$+$', the proof is similar in the case of~`$-$'.
  Call~$\sigma$ the permutation described by the biorder $(X,\qlt,<_X)$, and enumerate~$X$ as $\{x_1 <_X \dots <_X x_n\}$,
  so that $x_{\sigma(i)} \qlt x_{\sigma(j)}$ for all $i<j$.

  Given $x \not\in C$, define~$\phi(x) \in C$ as the smallest~$c \in C$ with regards to~$\qlt$ such that $x \to c$ is an edge, when such a~$c$ exists ($\phi(x)$ remains undefined otherwise).
  By definition of chain quasi-orderings, for~$x,y \not\in C$, we have $x \qlt y$ if and only if either $\phi(x) \qlt \phi(y)$, or~$\phi(x)$ is defined and~$\phi(y)$ is not.
  Thus, denoting $c_i \eqdef \phi(x_{\sigma(i)})$, to ensure $x_{\sigma(i)} \qlt x_{\sigma(j)}$, it must be that either $c_i \qlt c_j$ or~$c_i$ is defined and~$c_j$ is not.
  In particular, $c_i$ is well defined for all~$i < n$ and $c_i \qlt c_j$ for $i < j < n$.
  The choice of~$c_i$ then gives edges $c_i \from x_{\sigma(i)}$ and $c_j \to x_{\sigma(i)}$ for all $i<j<n$.
  Since~$c_n$ is not necessarily defined (in fact $\card{C} = n-1$ is possible), we discard $x_{\sigma(n)}$.

  By \cref{lem:perm-grid-minus-one}, restricting~$\sigma$ to~$[n-1]$ decreases the size of its grid by at most~$1$, leaving a grid of size~$(k^2+k)$.
  It consists of two partitions $A_1 < \dots < A_{k(k+1)}$ and $B_1 < \dots < B_{k(k+1)}$ of~$[n-1]$ into intervals, such that~$\sigma(A_i)$ intersects~$B_j$ for all~$i,j$.
  Define $C_i \eqdef \setst{c_\ell}{\ell \in A_i}$ and $X_i \eqdef \setst{x_\ell}{\ell \in B_i}$ the corresponding subsets of~$C,X$.
  We group them~$k+1$ by~$k+1$ as $C'_s \eqdef \bigcup_{i = (k+1)(s-1)+1}^{(k+1)s} C_i$ and similarly for~$X'_t$.
  To conclude, we only need to prove that the adjacency matrix of~$C'_s$ versus~$X'_t$ has rank at least~$k$ for all $s,t \in [k]$.

  For each~$i \in [k+1]$, there is by assumption some $\ell_i \in A_{(k+1)(s-1)+i}$ such that $\sigma(\ell) \in B_{(k+1)(t-1)+i}$.
  Thus $c_{\ell_i}$ is in~$C'_s$, and $x_{\sigma(\ell_i)}$ is in~$X'_t$ for all~$i \in [k+1]$.
  Also, by choice of~$c_i$, there are edges $c_{\ell_i} \from x_{\sigma(\ell_i)}$,
  and $c_{\ell_j} \to x_{\sigma(\ell_i)}$ whenever $\ell_j < \ell_i$.
  Thus, up to permutation of rows and columns, the matrix of $c_{\ell_1},\dots,c_{\ell_{k+1}}$ versus $x_{\sigma(\ell_1)}, \dots, x_{\sigma(\ell_{k+1})}$
  consists of a diagonal of~0, and a lower triangle filled with~1.
  Since this is a~$k+1$ by~$k+1$ matrix, this ensures that it has rank at least~$k$.
\end{proof}

We can now eliminate the second case of \cref{clm:chain-permutation2}.
\begin{proposition}\label{clm:chain-permutation3}
  For any permutation~$\sigma \in \Perm_n$, there exists some~$T \in \Tc$ with two disjoint ordered sets~$(A,<_A)$ and~$(B,<_B)$,
  such that the direction of edges inside~$A,B$ corresponds to~$<_A,<_B$ possibly up to reversal,
  and the adjacency matrix of~$(A,<_A)$ versus~$(B,<_B)$ is $\Mc_s(\sigma)$ for some encoding $s \in \matenc$.
\end{proposition}
\begin{proof}
  Given the desired permutation~$\sigma$ of size~$k$, consider~$t \eqdef f(k)$ given by \cref{thm:rank-division-extraction} ensuring some $\Mc_s(\sigma)$ inside any rank-$t$ division.
  Define then $g(k) = 2(t^2+t+1)$, and consider any permutation~$\tau$ with a $g(k)$-grid. Note in particular that~$\sigma$ is a pattern of~$\tau$.
  We apply \cref{clm:chain-permutation2} to~$\tau$, yielding one of two possible encoding of~$\tau$.
  In case~\ref{item:center-matrix} of \cref{clm:chain-permutation2}, we obtain exactly the desired structure only for~$\tau$ instead of~$\sigma$.
  Since~$\sigma$ is a pattern of~$\tau$, one can extract from it the same structure for~$\sigma$ instead.
  We may thus assume that case~\ref{item:side-chain-order} of \cref{clm:chain-permutation2} holds.
  That is, there is an ordered set~$(X,<_X)$ and a chain quasi-order~$\qlt$ which totally orders~$X$, and such that~$(X,\qlt,<_X)$ is the biorder representing~$\tau$ or~$\tau^{-1}$.

  Consider the chain~$C$ defining the chain quasi-order~$\qlt$.
  By definition, the order~$\qlt$ inside~$C$ coincides with the direction of edges up to reversal.
  Also, by assumption, the order~$<_X$ inside~$X$ coincides with edges up to reversal.
  Thus in~$C \cap X$, the orderings~$\qlt$ and~$<_X$ are either equal or opposite of each other.
  Defining $X' \eqdef X \setminus C$, \cref{lem:perm-grid-remove-chain} gives that the restricted permutation $(X',\qlt,<_X)$ still contains a grid half as large as that of $(X,\qlt,<_X)$, i.e.\ of size at least~$(t^2+t+1)$.

  We now apply \cref{lem:rank-minor-chain-order} to~$(C,\qlt)$ and $(X', <_X)$ to obtain a rank-$t$ division in their adjacency matrix.
  Then, \cref{thm:rank-division-extraction} yields a copy of~$\Mc_s(\sigma)$ in the same adjacency matrix for some encoding $s \in \matenc$, as desired.
\end{proof}

\subsection{Canonical obstructions}
In this section, we finally extract the obstructions of \cref{sec:obstructions}
from the structure described by \cref{clm:chain-permutation3}.
Here, rather than directly trying to obtain a specific permutation~$\sigma$,
we fix some $k \in \Nn$ and look for any permutation~$\sigma$ containing a $k$-grid.

\begin{proposition}\label{clm:chain-permutation4}
  Consider a tournament~$T$ with disjoint ordered sets of vertices~$(A,<_A)$ and~$(B,<_B)$,
  such that the direction of edges inside~$A,B$ corresponds to~$<_A,<_B$ possibly up to reversal,
  and the adjacency matrix of~$(A,<_A)$ versus~$(B,<_B)$ is~$\Mc_s(\sigma)$ for some encoding $s \in \matenc$ and permutation~$\sigma$ with a $(k+1)$-grid.
  Then~$T$ contains as subtournament $\Fc_{s'}(\sigma')$ for some $s' \in \{=,\le,\ge\}$ and some permutation~$\sigma'$ with a $k$-grid.
\end{proposition}
\begin{proof}
  Firstly, we can assume that the direction of edges inside~$A$ corresponds to the ordering~$<_A$ without reversal.
  Indeed, in the adjacency matrix between~$A$ and~$B$, reversing the ordering~$<_A$ (i.e.\ the ordering of columns) is equivalent to
  (1) pre-composing the permutation~$\sigma$ with $i \mapsto n-i+1$ (preserving the $(k+1)$-grid by \cref{lem:perm-grid-sym}),
  and (2) exchanging encodings~$\le_C$ and~$\ge_C$ (the other encodings are unaffected).
  Thus reversing~$<_A$ preserves the assumptions on the statement.
  A similar reasoning holds for~$<_B$, and we may thus assume that the direction of edges in~$A,B$ exactly matches with~$<_A$ and~$<_B$.

  Next, we can assume that~$s$ is one of~$=,\le_C,\ge_C$.
  If this is not the case, we swap the roles of~$A$ and~$B$.
  This corresponds to transposing the adjacency matrix (since~$A$ represented columns before and rows after), and complementing it, i.e.\ swapping~`0' for~`1' and vice versa (since a~`1' used to represent an edge from~$A$ to~$B$, and now the opposite).
  If the matrix was~$\Mc_{\neq}(\sigma)$ before, this yields~$\Mc_=(\sigma^{-1})$,
  and by \cref{lem:perm-grid-sym}, $\sigma^{-1}$ still has a $(k+1)$-grid.

  If the matrix was~$\Mc_{\le_R}(\sigma)$, the situation is slightly more complex:
  transposing and complementing yields $\Mc_{>_C}(\sigma^{-1})$,
  meaning the matrix defined as $\Mc_{\ge_C}(\sigma^{-1})$
  except for entries~$(i,\sigma^{-1}(i))$ that are~`0' instead of~`1'.
  \begin{claim}\label{clm:flipping-diagonal}
    For any permutation~$\sigma$ of size~$n$, there is a pattern~$\sigma'$ of~$\sigma$ of size~$n-1$
    such that~$\Mc_{\ge_C}(\sigma')$ is a submatrix of~$\Mc_{>_C}(\sigma)$.
  \end{claim}
  \begin{proofofclaim}
    Let~$\sigma'$ be the pattern of~$\sigma$ induced by~$\{1,\dots,n-1\}$.
    Its matrix~$M_{\sigma'}$ is obtained from~$M_\sigma$ by dropping column~$n$ and row~$\sigma(n)$.
    To obtain~$\Mc_{\ge_C}(\sigma')$ from~$\Mc_{>_C}(\sigma)$, it suffice to remove column~1 and row~$\sigma(n)$, both of which contain only~`0's.
    See \cref{fig:flipping-diagonal} for an illustration.
  \end{proofofclaim}
  \begin{figure}[htb]
    \begin{center}
      \begin{tikzpicture}[scale=0.3]
        \foreach \i in {0,...,8}{
          \foreach \j in {0,...,8}{
            \pgfmathtruncatemacro{\jq}{floor(\j/3)}
            \pgfmathtruncatemacro{\jr}{mod(\j,3)}
            \pgfmathtruncatemacro{\sj}{6-\jr*3+\jq}
            \pgfmathsetmacro{\c}{\i >= \sj ? "black" : "white"}
            \fill[\c] (\i,\j) rectangle +(1,1);
          }
        }
        \fill[white!70!red] (0,6) rectangle (9,7);
        \fill[black!50!red] (8,0) rectangle (9,9);
        \draw[white!50!black] (0,0) grid (9,9);

        \begin{scope}[xshift=10cm]
        \foreach \i in {0,...,8}{
          \foreach \j in {0,...,8}{
            \pgfmathtruncatemacro{\jq}{floor(\j/3)}
            \pgfmathtruncatemacro{\jr}{mod(\j,3)}
            \pgfmathtruncatemacro{\sj}{6-\jr*3+\jq}
            \pgfmathsetmacro{\c}{\i > \sj ? "black" : "white"}
            \fill[\c] (\i,\j) rectangle +(1,1);
          }
        }
        \fill[white!70!red] (0,6) rectangle (9,7);
        \fill[white!70!red] (0,0) rectangle (1,9);
        \draw[white!50!black] (0,0) grid (9,9);
        \end{scope}
      \end{tikzpicture}
    \end{center}
    \caption{
      Proof of \cref{clm:flipping-diagonal}.
      Left: the matrix~$\Mc_{\ge_C}(\sigma)$ for $\sigma = 369258147$.
      Removing the 9th column and the row with index $\sigma(9) = 7$, marked in red, yields~$\Mc_{\ge_C}(\sigma')$ where $\sigma' = 36258147$ is a pattern of~$\sigma$.\\
      Right: the matrix~$\Mc_{>_C}(\sigma)$.
      Removing again row~$7$ and this time the 1st column, both of which contain only~`0's, yields exactly the same submatrix~$\Mc_{\ge_C}(\sigma')$.
    }
    \label{fig:flipping-diagonal}
  \end{figure}
  Applying \cref{clm:flipping-diagonal} to the matrix~$\Mc_{>_C}(\sigma^{-1})$
  yields~$\Mc_{\ge_C}(\sigma')$ as submatrix, where~$\sigma'$ is a pattern of~$\sigma^{-1}$ whose size is only one less than that of~$\sigma^{-1}$.
  Then, by \cref{lem:perm-grid-minus-one}, $\sigma'$ contains a grid only one smaller than that of~$\sigma^{-1}$, i.e.~$\sigma'$ has a $k$-grid.
  Thus the case of~$\Mc_{\le_R}(\sigma)$ where~$\sigma$ has a $(k+1)$-grid, reduces to that of~$\Mc_{\ge_C}(\sigma')$ where~$\sigma'$ has a $k$-grid.
  The reduction from~$\ge_R$ to~$\le_C$ is similar.

  Finally, we are in the situation where we have two disjoint chains~$A,B$ ordered by the direction of edges (without reversal),
  and the adjacency matrix between~$A$ and~$B$ is one of~$\Mc_=(\sigma')$, $\Mc_{\le_C}(\sigma')$, or~$\Mc_{\ge_C}(\sigma')$ for some permutation~$\sigma'$ with a $k$-grid.
  This exactly corresponds to the definition of~$\Fc_=(\sigma')$, $\Fc_\le(\sigma')$, and~$\Fc_\ge(\sigma')$ respectively.
\end{proof}

Combining \cref{clm:chain-permutation3,clm:chain-permutation4} gives that for all~$k$, there is a permutation~$\sigma$ with a $k$-grid and an encoding $s \in \{=,\le,\ge\}$ such that $\Fc_s(\sigma)$ is contained in some tournament $T \in \Tc$, and thus $\Fc_s(\sigma) \in \Tc$ too as~$\Tc$ is hereditary.
One of the three encodings~$=,\le,\ge$ must occur for infinitely many values of~$k$,
i.e.\ there is some $s \in \{=,\le,\ge\}$ such that for arbitrary large~$k \in \Nn$, there is a permutation~$\sigma$ with a $k$-grid such that $\Fc_s(\sigma) \in \Tc$.
Any permutation~$\tau$ of size~$k$ is a pattern of such a~$\sigma$, which by \cref{lem:forb-pattern} gives that $\Fc_s(\tau)$ is contained in~$\Fc_s(\sigma)$.
It follows that~$\Fc_s(\tau) \in \Tc$ for all permutation~$\tau$, that is $\Fc_s \subseteq \Tc$, proving \cref{thm:forbidden-tournaments}.

\section{Dense oriented graphs and relational structures}\label{sec:generalisation}
In this final section, we generalise our main result in two directions:
we replace tournaments by oriented graphs with bounded independent sets,
and consider relational structures formed from such oriented graphs augmented by arbitrary binary relations.

\subsection{Generalised binary search trees}
We consider oriented graphs~$D$ (forbidding digons).
The independence number~$\alpha(D)$ is the maximum size of an independent set (subset of vertices with no edges between them) in~$D$.
An oriented graph~$D$ is a tournament if and only if~$\alpha(D) = 1$.
In a sense, oriented graphs whose independence number is bounded by a small constant resemble tournaments;
we will generalise our results to them.

We first generalise the notion of BST to arbitrary oriented graphs.
A BST on an oriented graph~$D$ is a rooted ordered \emph{ternary} tree~$S$
(meaning that each node has a left, centre, and right child, any of which may be missing)
satisfying the following condition.
For any node~$x$ of the tree, let~$L_x,C_x,R_x$ denote the left, centre, and right subtrees of~$x$.
Then
\[ L_x \subseteq N^-(x) \qquad C_x \cap N(x) = \emptyset \qquad L_r \subseteq N^+(x). \]
The left-to-right order~$<_S$ corresponding to the tree is the order satisfying for all node~$x$
\[ L_x <_S x <_S C_x <_S R_x. \]
The choice of placing~$C_x$ after and not before~$x$ here is arbitrary.

Any branch~$B$ of the tree is an acyclic subgraph in which the direction of edges matches~$<_S$.
However, it is not necessarily a chain.
Indeed, if~$x \in B$ is such that the branch~$B$ descends into the central subtree of~$x$,
then~$x$ is non-adjacent to all its descendants in~$B$.
Nonetheless, $B$ can be transformed into a chain at a limited cost when~$\alpha(D)$ is small:
Let~$X \subset B$ be the set of all such centre-branching nodes in~$B$.
Then~$X$ is an independent set and thus~$\card{X} \le \alpha(D)$,
while the rest~$B \setminus X$ of the branch is a chain.

Next, we generalise chain quasi-orders.
For~$C$ a chain in~$D$ enumerated as~$C = \{c_1,\dots,c_k\}$ such that~$c_i \to c_j$ iff~$i < j$,
we define
\[ A_i = \bigcap_{j \le i} N^+(c_i)
  \qquad \text{and} \qquad B_i = A_{i-1} \setminus (N^+(c_i) \cup \{c_i\}). \]
In other words, $A_{i-1}$ is partitioned as $A_i \uplus \{c_i\} \uplus B_i$,
where~$A_i$ contains the out-neighbours of~$c_i$, and~$B_i$ the in-neighbours and non-neighbours.
The chain quasi-order~$\qle_C^+$ is then defined by
\[ B_1 \qlt_C^+ c_1 \qlt_C^+ B_2 \qlt_C^+ c_2 \qlt_C^+ \dots B_k \qlt_C^+ c_k \qlt_C^+ A_k, \]
with $B_1,\dots,B_k$ and~$A_k$ being the equivalence classes.
That is, we order a node~$x$ according to the smallest~$i$ such that~$x \not\in N^+(c_i)$.
The dual quasi-order~$\qle_C^-$ is defined in the same way, but reversing the direction of all edges.

We can now generalise partition extraction (\cref{lem:bst-partition-extraction}) to oriented graphs with bounded~$\alpha$.
\begin{lemma}
  \label{lem:bst-partition-extraction2}
  There is a function~$f(k,\alpha) = 2^{O(k+\alpha)}$ such that
  for any oriented graph~$D$ with independence number~$\alpha$,
  BST ordering~$<_S$ of~$D$, and family~$\Pc$ of at least~$f(k,\alpha)$ disjoint intervals of~$<_S$,
  one can find a chain quasi-ordering~$\qle$ and a subfamily~$\Pc' \subset \Pc$ such that
  $\card{\Pc'} \ge k$ and such that the parts of~$\Pc'$ are non-overlapping for~$\qle$.

  Furthermore, $\Pc'$ as well as the chain~$C$ and orientation~$o \in \{+,-\}$ defining~$\qle$ can be computed in linear time.
\end{lemma}
\begin{proofsketch}
  The proof of \cref{lem:bst-partition-extraction} can be applied to directed graphs with essentially no modification.
  The only issue is that it does not yield a chain quasi-order,
  because the branches of~$S$ are not themselves chains.
  Instead, the quasi-order~$\qle$ compatible with the resulting partition~$\Pc'$
  can be described by the following structure (up to inversion of all edges).
  There is a sequence~$C = \{c_1,\dots,c_k\}$ of vertices.
  Each~$c_i$ has a \emph{type}: either \emph{complete} or \emph{anti-complete}.
  If~$c_i$ is complete, then for all~$j > i$ there is an edge~$c_i \to c_j$.
  If~$c_i$ is anti-complete, then for all~$j > i$ there is no edge between~$c_i$ and~$c_j$.
  We then define subsets~$A_i$ by induction: $A_0$ is $V(T) \setminus C$,
  and~$A_i$ is defined as~$A_{i-1} \cap N^+(c_i)$ when~$c_i$ is complete
  or as~$A_{i-1} \setminus N(c_i)$ when~$c_i$ is anti-complete.
  Then~$B_i$ is defined as~$A_{i-1} \setminus A_i$, and the quasi-order is
  \[ B_1 \qlt_C^+ c_1 \qlt_C^+ B_2 \qlt_C^+ c_2 \qlt_C^+ \dots B_k \qlt_C^+ c_k \qlt_C^+ A_k. \]
  When there are no anti-complete nodes in~$C$, this structure is exactly a chain quasi-order.

  Remark now that the anti-complete nodes in~$C$ induce an independent set.
  It follows that there are at most~$\alpha$ anti-complete nodes.
  Then, we can simply remove any anti-complete node~$c_i$, and remove from~$\Pc'$ any part which intersects~$B_i$,
  at the cost of at most~$\alpha$ parts of~$\Pc'$.
\end{proofsketch}

\subsection{Relational structures over tournaments}
Let us now focus on the second generalisation:
to a tournament~$T = (V,E)$, we add arbitrary binary relations~$R_1,\dots,R_\ell \subseteq V^2$,
yielding a relational structure~$(V,E,R_1,\dots,R_\ell)$.
We will show that the tournament~$(V,E)$ in such structures is enough to obtain our main results,
without any hypothesis on~$R_1,\dots,R_\ell$.
Remark that when the tournament~$(V,E)$ is transitive,
this exactly corresponds to the ordered structures considered in \cite{twin-width4}.
The relations~$R_i$ are required to be binary for twin-width to be well-defined.
Naturally, $(V,E)$ can also be an oriented graph with bounded independence number instead of a tournament.

Formally, we fix~$\Sigma = \{E,R_1,\dots,R_\ell\}$ a binary relational signature with a distinguished relation~$E$,
and consider~$\Sigma$-structures in which~$E$ is interpreted
as an oriented graph with independence number bounded by some constant~$\alpha$.

A chain order representation of the matrix~$M$ in a $\Sigma$-structure~$S$ consists of the following:
\begin{itemize}
  \item two disjoint ordered subsets~$(A,\qle^A)$ and~$(B,\qle^B)$,
    where~$\qle^A,\qle^B$ are two chain quasi-orderings in~$(V,E)$, and
  \item a relation~$R$ among $E,R_1,\dots,R_\ell$ such that
    the adjacency matrix of~$R$ restricted to~$(A,\qle^A)$ versus~$(B,\qle^B)$ is~$M$.
\end{itemize}
If~$\sigma$ is a permutation, we define a chain order representation of~$\sigma$ in~$S$
as being a chain order representation of the matrix~$\Mc_s(\sigma)$ in~$S$
where~$s$ is any of the encodings in $\matenc = \{=,\neq,\le_R,\ge_R,\le_C,\ge_C\}$, see~\cref{sec:perm-matrice-encoding}.

Using \cref{lem:bst-partition-extraction2},
it is straightforward to generalise \cref{thm:forbidden-tournaments-intro}.
\begin{theorem}
  \label{thm:chain-order-extraction-gen}
  Let~$\Sc$ be a hereditary class of $\Sigma$-structures in which~$E$ is interpreted
  as oriented graphs with bounded independence number.
  If there are structures~$S \in \Sc$ with BST orders~$<$ for which~$\tww(S,<)$ is arbitrarily large,
  then~$\Sc$ contains chain order representations of all permutations.
\end{theorem}
\begin{proofsketch}
  The proof is the same as the first half of \cref{thm:forbidden-tournaments-intro},
  up to \cref{clm:chain-permutation1}:
  Given~$S = (V,E,R_1,\dots,R_\ell)$ and the BST ordering~$<$ such that~$\tww(S,<)$ is large,
  \cref{thm:grid-theorem} (or rather its natural generalisation to ordered binary structures,
  see Theorem~7 and Lemma~24 in~\cite{twin-width4}) yields a high-rank division
  in the adjacency matrix of one of the relations~$E,R_1,\dots,R_\ell$, with rows and columns ordered by~$<$.
  \Cref{lem:bst-partition-extraction2} and the arguments of \cref{sec:approx}
  then allow replacing~$<$ with chain orderings,
  yielding a chain order representation of a matrix with a high-rank division.
  Finally, \cref{thm:rank-division-extraction} extracts an encoding~$\Mc_s(\sigma)$ for any permutation~$\sigma$ of size~$k$,
  with~$k$ an unbounded function of~$\tww(S,<)$.
  Since the latter can be arbitrarily large in~$\Sc$, one obtains chain order representations of all permutations.
\end{proofsketch}
A variant of \cref{lem:chain-order-obstruction} shows that
structures containing chain order representations of all permutations have unbounded twin-width.
Following the arguments of \cref{sec:approx},
this yields an \fpt approximation algorithm for twin-width of such relational structures:
\begin{theorem}
  There are functions~$f,g$ and an algorithm which given a structure $S = (V,E,R_1,\dots,R_\ell)$
  with twin-width~$t$, where~$(V,E)$ is an oriented graph with independence number~$\alpha$,
  computes a contraction sequence of width~$f(t,\alpha,\ell)$ in time~$g(t,\alpha,\ell) \cdot \card{V}^{O(1)}$.
  \label{thm:tww-apx-gen}
\end{theorem}

To obtain the full list of equivalent conditions,
we additionally need to extract canonical obstructions from chain order representations
by generalising the proofs of \cref{sec:extraction}.

Consider a chain order representation of a permutation~$\sigma$ on~$n$ elements.
It consists of four subsets of vertices:
$A,B$ between which the relation~$R$ induces the matrix~$\Mc_s(\sigma)$,
and the chains~$C_A,C_B$ which define the appropriate chain orderings~$\qle^A,\qle^B$.
The proof of \cref{clm:chain-permutation2} using \cref{lem:perm-ramsey} applies with minor modifications to reduce to one of two cases:
\begin{enumerate}
  \item \label{item:center-matrix-gen}
    either the orientation of edges of~$E$ in~$A,B$ coincides with~$\qle^A,\qle^B$ respectively up to reversal,
    and the adjacency matrix of~$(A,\qle^A)$ versus~$(B,\qle^B)$ for the relation~$R$ is~$\Mc_s(\sigma)$,
  \item \label{item:side-chain-order-gen}
    or the orientation of edges of~$E$ in~$A$ defines an ordering~$<_A$ such that the biorder~$(A,\qle_A,<_A)$ represents the permutation~$\sigma$, or similarly with~$B$.
\end{enumerate}
The only difference with the statement of \cref{clm:chain-permutation2} is that in case~\ref{item:center-matrix-gen}, the matrix~$\Mc_s(\sigma)$ can be encoded by any of the relations~$R \in \Sigma$ instead of just~$E$.
A subtlety is that since~$(V,E)$ is now an oriented graph and not a tournament,
applying \cref{lem:perm-ramsey} could also lead to~$A$ or~$B$ being an independent set instead of a chain.
However, since the size of independent sets is restricted, this cannot happen as long as sufficiently large permutations are considered.

Next, the proofs of \cref{lem:rank-minor-chain-order,clm:chain-permutation3} can be adapted with minor changes accounting for non-edges,
to show that case~\ref{item:side-chain-order-gen} above reduces to case~\ref{item:center-matrix-gen}.
Thus, for any permutation~$\sigma$, we can find in some~$S \in \Sc$
two disjoint chains~$(A,<_A)$ and~$(B,<_B)$ where~$<_A,<_B$ coincides with the orientation of edges of~$E$ inside~$A,B$ up to reversal, such that the adjacency matrix of some relation~$R \in \Sigma$ between~$(A,<_A)$ and~$(B,<_B)$ is $\Mc_s(\sigma)$, for some encoding~$s \in \matenc$.
Using the same arguments as \cref{clm:chain-permutation4}, we can assume that the orientation of edges of~$E$ in fact coincides with~$<_A,<_B$ without reversal.

To summarise, for any permutation~$\sigma$, we obtain the following structure in some~$S \in \Sc$ for some encoding~$s \in \matenc$:
disjoint sets $A = \{a_1 <_A \dots <_A a_n\}$ and $B = \{b_1 <_B \dots <_B b_n\}$, and edges
\begin{align}
  & \text{$a_i \to a_j$ and $b_i \to b_j$ in~$E$ for all $i<j$, and} \label{eq:side-orders} \\
  & \text{$a_ib_j \in R$ if and only if there is a~1 at position~$(i,j)$ in~$\Mc_s(\sigma)$.} \label{eq:matrix}
\end{align}
In the case of tournaments, $E = R$ was the only relation in~$\Sigma$ and the above entirely specified the tournament induced by~$A \cup B$, describing the minimal obstructions from \cref{sec:obstructions}.
This is no longer the case: edges of~$E$ between~$A$ and~$B$, of~$R$ inside~$A$ and~$B$, and of all the remaining relations are left unspecified by the previous description.
To obtain canonical obstructions, and particularly to obtain an efficient interpretation from~$\Sc$ to the class of all graphs, we need to eliminate these remaining choices.
To this end, we will once again use \cref{lem:perm-ramsey}.

The setup is similar to the proof of \cref{clm:chain-permutation2}.
Given a structure as above representing~$\tau$,
we consider the biorder $\Oc_\tau = ([n],<,<_\tau)$ and define the colouring $\lambda : [n]^2 \to \Gamma$
where~$\lambda(i,j)$ the substructure induced by $a_i,a_j,b_{\tau(i)},b_{\tau(j)}$.
Here, the set~$\Gamma$ of colours is the set of all $\Sigma$-structures on four vertices.
Given now any desired permutation~$\sigma$, consider~$\tau$ function of~$\sigma$ and~$|\Gamma|$
so that \cref{lem:perm-ramsey} gives a subset $X \subset [n]$ such that
the substructure~$\Oc_\tau[X]$ is isomorphic to~$\Oc_\sigma$,
and~$\lambda$ restricted to~$X$ only depends on the orders~$<$ and~$<_\tau$,
i.e.\ there is a function $\eta : \{-1,0,1\}^2 \to \Gamma$ such that $\lambda(i,j) = \eta(\ot(i,j), \ot_\tau(i,j))$.
Restricting to $A' = \setst{a_i}{i \in X}$ and $B' = \setst{b_{\sigma(i)}}{i \in X}$,
we obtain an encoding of~$\sigma$ in which the substructure induced by $a_i,a_j,b_{\tau(i)},b_{\tau(j)}$ is entirely specified by~$\eta(\ot(i,j), \ot_\tau(i,j))$.

Thus, we now have an encoding of any given permutation~$\sigma$
which in addition to conditions~\cref{eq:side-orders,eq:matrix} above, satisfies
\begin{equation}
  \label{eq:deterministic}
  \text{\parbox{0.85\textwidth}{%
  There is a map $\eta : \{-1,0,1\}^2 \to \Gamma$ such that
    the adjacencies of vertices $a_i,a_j,b_{\sigma(i)},b_{\sigma(j)}$ are as specified by $\eta(\ot(i,j), \ot_\sigma(i,j))$.}}
\end{equation}
The choice of~$\eta$ and~$\sigma$ entirely specifies this structure, which we call~$\Fc_\eta(\sigma)$.
Some choices of~$\eta$ are impossible, as they can either be self contradictory (there is some redundant information in this encoding), or contradict conditions~\cref{eq:side-orders,eq:matrix}.
Denote by $\Hc \subset \{-1,0,1\} \to \Gamma$ the set of valid choices of~$\eta$.
Each $\eta \in \Hc$ yields an encoding~$\Fc_\eta(\sigma)$ of any given permutation~$\sigma$,
and we call~$\Fc_\eta$ the hereditary closure of all~$\Fc_\eta(\sigma)$.

From \cref{thm:chain-order-extraction-gen} and the previous arguments, we obtain the following:
\begin{lemma}
  \label{lem:chain-order-extraction-gen}
  Let~$\Sc$ be a hereditary class of $\Sigma$-structures in which~$E$ is interpreted
  as an oriented graph with independence number bounded by~$\alpha$.
  If~$\Sc$ has unbounded twin-width, then it contains as a subset~$\Fc_\eta$ for some~$\eta \in \Hc$.
\end{lemma}

We finally need to show that~$\Fc_\eta$ is an obstruction to the properties we are interested in.
Efficiently interpreting all graphs is the most technical step.
\begin{lemma}
  \label{lem:forb-gen-interpret-perm}
  Each of the classes of obstructions~$\Fc_\eta$ for $\eta \in \Hc$ efficiently interprets the class~$\Oc_\Perm$ of all biorders.
\end{lemma}
\begin{proofsketch}
  The technique is the same as in the case of tournaments in \cref{lem:forb-interpret-perm}.
  Consider some encoding~$\Fc_\eta(\sigma)$ of an arbitrary permutation~$\sigma$, with vertex set~$A \uplus B$
  as described by conditions~\cref{eq:side-orders,eq:matrix,eq:deterministic}.
  If given a way to distinguish vertices in~$A$ from vertices in~$B$,
  then one can order~$A$ and~$B$ according to the edges of~$E$, and read the matrix~$\Mc_s(\sigma)$ between them in the relation~$R$.
  It is then simple to reconstruct the biorder~$\Oc_\sigma$ with a first-order interpretation with the same method as \cref{lem:forb-interpret-perm} (see \cite[Lemma~40]{twin-width4}).

  The difficulty is thus to distinguish~$A$ from~$B$.
  There are several cases to consider depending on the edges of~$E$ specified by~$\eta$ between~$A$ and~$B$.
  In each case, we construct a slightly larger structure~$S \in \Fc_\eta$, and show that from~$S$,
  one can interpret~$\Fc_\eta(\sigma)$ augmented by unary predicates to distinguish~$A$ and~$B$.

  The rest of this proof focuses on~$E$, and not the other relations $R_1,\dots,R_\ell$.
  In a structure~$\Fc_\eta(\sigma')$, we enumerate vertices as $A = \{a_1 <_A a_2 <_A \dots\}$ and $B = \{b_1 <_B b_2 <_B \dots \}$
  as in the definition of~$\Fc_\eta(\sigma')$ above.

  \subsubsection*{When there are many non-edges}
  In~$\Fc_\eta(\sigma)$, for all choices of~$i,j \in [n]$ satisfying $i < \sigma^{-1}(j)$ and $\sigma(i) < j$,
  the presence of an edge between~$a_i$ and~$b_j$ and its direction depend only on~$\eta$.

  First, let us assume that~$a_ib_j$ is a non-edge for all such~$i,j$.
  Then, we embed~$\sigma$ in a permutation~$\sigma' \in \Perm_{n+2}$ by setting $\sigma'(1) = 1$, $\sigma'(n+2) = n+2$, and $\sigma'(i+1) = \sigma(i)+1$.
  In~$\Fc_\eta(\sigma')$, we have that~$a_1b_j$ is a non-edge for all~$j>1$, and so is~$a_ib_{n+2}$ for all~$i<n+2$.
  Also, the non-neighbourhood of any~$a_i$ (resp.~$b_j$) is contained in~$B$ (resp.~$A$).
  Delete now the vertices~$a_{n+2}$ and~$b_1$. The resulting structure is still in~$\Fc_\eta$, as it is a hereditary class.
  Now $A$ and~$B$ are the non-neighbourhoods of~$b_{n+2}$ and~$a_1$, and are the only pair of non-neighbourhoods which cover the entire vertex set.
  This uniquely defines~$A$ and~$B$, and can be expressed in first-order logic.
  It then suffices to forget~$a_1$ and~$b_{n+2}$ (the minimum of~$A$ and maximum of~$B$) to retrieve~$\Fc_\eta(\sigma)$ while knowing the bipartition~$A,B$.

  Observe that the structure described above is entirely determined by~$\eta$ and~$\sigma$, and is easily computed in polynomial time.
  Thus we obtain an efficient interpretation of~$\Oc_\Perm$ from~$\Fc_\eta$.

  Similar arguments apply when there are non-edges between~$a_i$ and~$b_j$ for the other three combinations of inequalities between~$i,\sigma^{-1}(j)$, and~$\sigma(i),j$: one only needs to swap the roles of~$a_1$ with~$a_{n+2}$, and~$b_1$ with~$b_{n+2}$ as required.
  We can thus assume that there is an edge in~$E$ between~$a_i$ and~$b_j$ for all~$i,j$ except possibly $j = \sigma(i)$.

  \subsubsection*{When there is an anti-matching}
  Assume now that~$a_ib_{\sigma(i)}$ is a non-edge for all~$i$.
  Without loss of generality, assume that whenever~$i,j$ satisfy $i < \sigma^{-1}(j)$ and $j < \sigma(i)$, then $a_i \to b_j$ is an edge.
  (in the opposite case, the roles of~$A$ and~$B$ can be swapped in the following argument).
  Define now~$\sigma' \in \Perm_{n+1}$ by $\sigma'(i) = \sigma(i)+1$ and $\sigma'(n+1) = 1$.
  Thus~$a_i \to b_1 \to b_i$ are edges for all~$i > 1$.
  Delete now~$a_{n+1}$, so that~$b_1$ is the unique vertex without a non-neighbour, and the partition~$A,B$ is given by the in-neighbourhood and out-neighbourhood of~$b_1$.
  This allows to describe~$A,B$ in first-order logic.
  Once again, the previous structure can be constructed in polynomial time from~$\sigma$, hence~$\Fc_\eta$ efficiently interprets~$\Oc_\Perm$.

  These two first cases deal with all possible non-edges in~$E$ between~$A$ and~$B$.
  Thus, we can now assume~$E$ to be a tournament, and not just an oriented graph.

  \subsubsection*{When~$E$ is a total order}
  Consider now the case where~$\eta$ specifies that edges in~$E$ are always oriented from~$A$ to~$B$, or vice versa.
  Then the relation~$E$ is simply a total ordering, making~$\Fc_\eta$ a class of ordered structures with unbounded twin-width.
  In this case, we can invoke the known results on ordered structures to conclude~\cite[Theorem~42]{twin-width4}.

  \subsubsection*{When~$E$ has large twin-width}
  We are only left with the case where~$E$ is a tournament, and the edges between~$A$ and~$B$ are not all oriented in the same direction.
  Consider the permutation~$\tau = 132$.
  One can pick~$i,j \in \{1,2,3\}$ to realise all possible pairs of relative orderings of~$i,\tau^{-1}(j)$ on the one hand, and~$\tau(i),j$ on the other.
  Since the direction of edges between~$A,B$ only depends on these relative orderings, and both directions must be feasible, it follows that~$\Fc_\eta(\tau)$ must contain edges going in both directions between~$A$ and~$B$.

  Now consider a permutation~$\sigma \in \Perm_n$ with a sufficiently large grid,
  and consider the substitution of~$\tau$ in~$\sigma$,
  that is the permutation $\sigma' \in \Perm_{3n}$ defined by
  \[ \sigma'(3(i-1)+j) = 3(\sigma(i)-1) + \tau(j). \]
  In~$\Fc_\eta(\sigma')$, observe that the sets $\{a_{3i+1},a_{3i+2},a_{3i+3}\}$ and $\{b_{3j+1},b_{3j+2},b_{3j+3}\}$ are non-homogeneous if and only if $j = \sigma(i)$.
  This implies that the adjacency matrix of~$(A,<_A)$ versus~$(B,<_B)$ for the relation~$E$ contains a high rank division.
  In other words, the tournament~$(A \cup B, E)$ itself has large twin-width, without having to consider the other relations~$R_1,\dots,R_\ell$.

  Thus we can use the results of \cref{sec:extraction} to reduce to the case where in~$\Fc_\eta(\sigma)$,
  the tournament~$(A \cup B,E)$ is one of the obstruction $\Fc_=(\sigma)$, $\Fc_\le(\sigma)$, $\Fc_\ge(\sigma)$ described in \cref{sec:obstructions}.
  We then conclude using the efficient interpretation from tournament obstructions given by \cref{lem:forb-interpret-perm}.
\end{proofsketch}

The other hardness properties for~$\Fc_\eta$ are easily obtained.
\begin{lemma}
  \label{lem:chain-order-obstruction-gen}
  Each of the classes of obstructions~$\Fc_\eta$, $\eta \in \Hc$
  \begin{enumerate}[noitemsep]
    \item has unbounded twin-width,
    \item contains at least~$\frac{(n/2)!}{c^n}$ structures on~$n$ vertices counted up to isomorphism
      for some constant~$c$,
    \item efficiently interprets the class of all graphs,
    \item and has an \aw-hard FO model checking problem.
  \end{enumerate}
\end{lemma}
\begin{proof}
  For each~$\Fc_\eta$ and any permutation~$\sigma \in \Perm_n$,
  there is a $\Sc \in \Fc_\eta$ consisting of two chains~$A,B$ of length~$n$ each,
  encoding~$\sigma$ as described by conditions~\cref{eq:side-orders,eq:matrix}.
  After guessing the two subsets~$A,B$ (at most~$2^{2n}$ choices), one can deterministically reconstruct~$\sigma$.
  This implies that there are at least~$\frac{n!}{4^n}$ non-isomorphic structures on~$2n$ vertices in~$\Fc_\eta$, proving~(2).
  It follows by \cref{thm:small} that~$\Fc_\eta$ has unbounded twin-width~(1).
  Finally, the efficient interpretation~(3) is provided by \cref{lem:forb-gen-interpret-perm,lem:perm-inter-graph}, 
  and implies model checking hardness~(4) by \cref{lem:interpretation-reduction}.
\end{proof}

The main result follows from \cref{lem:chain-order-extraction-gen,lem:chain-order-obstruction-gen}.
\begin{theorem}
  Let~$\Sc$ be a class of $\Sigma$-structures in which~$A$ is interpreted
  as oriented graphs with independence number at most some constant~$k$.
  Under the assumption $\fpt \neq \aw$, the following are equivalent:
  \begin{enumerate}[noitemsep]
    \item $\Sc$ has bounded twin-width,
    \item FO model checking in~$\Sc$ is \fpt,
    \item FO model checking in~$\Sc$ is not \aw-hard,
    \item $\Sc$ does not interpret the class of all graphs,
    \item $\Sc$ is monadically NIP,
    \item $\Sc$ contains at most~$c^n$ structures on~$n$ vertices for some constant~$c$, counted up to isomorphism.
  \end{enumerate}
\end{theorem}

\section*{Acknowledgements}
The authors are extremely grateful to Édouard Bonnet for stimulating discussions on the topics of this work,
and Szymon Toruńczyk for helpful explanation on model theoretic aspects of \cref{sec:approx}.
We would also like to thank the anonymous reviewers for their very helpful remarks and suggestions.

\printbibliography

\end{document}